\documentclass[10pt]{IEEEtran}
\usepackage{url}
\usepackage[dvips]{graphicx}
\usepackage{amssymb,amsfonts,amsmath,amsthm,amscd,dsfont,mathrsfs}
\usepackage{graphicx,float,psfrag,epsfig,color}
\usepackage{subfig}
\usepackage{cite}

\newcommand\independent{\protect\mathpalette{\protect\independent}{\perp}} 
\def\independent#1#2{\mathrel{\rlap{$#1#2$}\mkern2mu{#1#2}}}

\newcommand{\mR}{\mathbb{R}} 
\newcommand{\mN}{\mathbb{N}}
\newcommand{\mZ}{\mathbb{Z}}

\newcommand{\pp}{\mathbb{P}}
\newcommand{\E}{\mathbb{E}}
\newcommand{\cF}{{\cal F}}

\newcommand{\bracket}[1]{\ensuremath{\left\langle#1\right\rangle}}

\newcommand{\rank}{\mathrm{rank}}

\newcommand{\lin}{{\cal L}}
\newcommand{\lino}{{\cal L}(\Omega,\cF,\pp)}

\newcommand{\mmse}{{\text{mmse}}}
\newcommand{\vecu}{\boldsymbol{u}}
\newcommand{\vecv}{\boldsymbol{v}}
\newcommand{\vecx}{\boldsymbol{x}}
\newcommand{\vecy}{\boldsymbol{y}}

\newcommand{\col}{\mathsf{C}}
\newcommand{\row}{\mathsf{R}}
\newcommand{\normal}{\mathsf{N}}
\newcommand{\var}{{\text{Var}}}

\theoremstyle{definition}
\newtheorem{definition}{Definition}
\theoremstyle{plain}
\newtheorem{thm}{Theorem}
\theoremstyle{plain}
\newtheorem{prop}{Proposition}
\theoremstyle{plain}
\theoremstyle{plain}
\theoremstyle{plain}
\theoremstyle{thm}
\newtheorem{remark}{Remark}
\theoremstyle{discussion}
\theoremstyle{plain}

\begin{document}

\title{Multi Terminal Probabilistic Compressed Sensing}
\author{Saeid Haghighatshoar
\\ EPFL, Lausanne, Switzerland\\
Email: saeid.haghighatshoar@epfl.ch}
\maketitle 

\begin{abstract}
In this paper, the \textit{`Approximate Message Passing'} (AMP) algorithm, initially developed for compressed sensing of signals under i.i.d. Gaussian measurement matrices, has been extended to a multi-terminal setting (MAMP algorithm). It has been shown that similar to its single terminal counterpart, the behavior of MAMP algorithm is fully characterized by a \textit{`State Evolution'} (SE) equation for large block-lengths. This  equation has been used to obtain the rate-distortion curve of a multi-terminal memoryless source. It is observed that by spatially coupling the measurement matrices, the rate-distortion curve of MAMP algorithm undergoes a phase transition, where the measurement rate region corresponding to a low distortion (approximately zero distortion) regime is fully characterized by the joint and conditional R\'enyi information dimension (RID) of the multi-terminal source. This measurement rate region is very similar to the rate region of the Slepian-Wolf distributed source coding problem where the RID plays a role similar to the discrete entropy.  

Simulations have been done to investigate the empirical behavior of MAMP algorithm. It is observed that simulation results match very well with predictions of SE equation for reasonably large block-lengths.
\end{abstract}

\begin{IEEEkeywords}
Approximate message passing (AMP), Gaussian measurement matrices, Spatial coupling, Multi-Terminal Approximate Message Passing (MAMP), R\'enyi information dimension, Multi-terminal (distributed) compressed sensing.
\end{IEEEkeywords}

\section{Introduction}
Let $(x^n,y^n)$ be a realization of a two terminal memoryless sources $(X,Y)$ with a probability distribution $p_{X,Y}$ over $\mR^2$ and assume that one is  interested to recover the signal in both terminals $(x^n, y^n)$ by taking sufficiently many linear measurements $\vecu=A x^n$ and $\vecv=B y^n$, where $A$ and $B$ denote the measurement matrices in $T_X$ and $T_Y$ respectively. In particular, it is implicitly assumed that the measurements are taken separately from each terminal whereas for the recovery, one has access to the measurements $(\vecu,\vecv)$ from both terminals. 

This problem in its general multi-terminal form is ubiquitous in different distributed processing systems and specially in ad hoc sensor networks where a collection of sensors measure a distributed environmental signal like temperature, humidity, etc. One can imagine a particular sensor as a terminal which takes a collection of linear measurements and transmits the gathered data to a data fusion center by routing them via the other sensors. Because of limited communication and low processing power of sensors, it is difficult to take joint measurements from two or several different terminals even if they are very closed to one another. Therefore, one can reasonably assume that the measurements are taken separately from each terminal and processed jointly in a data fusion center to recover the distributed signal. Usually there is a high correlation among terminals and one can exploit this redundancy to reduce the required number of measurements. In particular, in a very low energy scenario like a sensor network this results in a saving in the energy consumption of  devices which in turn, increases the life time of the network. 

There are two different kinds of correlation that should be considered: temporal and spatial. In a scenario like sensor networks, temporal correlations result because of the slow changes of the natural phenomenon like temperature, humidity, etc. Temporal correlations usually can be moderated by suitable sampling time and preprocessing of the signal before transmission. Spatial correlations are more important and much more difficult to deal with. If the sensors are densely distributed in the environment for precise data acquisition, the resulting measurements from different terminals will be highly  redundant thus the network energy resources are  wasted without any significant gain. Therefore, it is always desirable to reduce the number of sensors to a minimum possible and still be able to recover the environmental distributed signal. Compared with a densely distributed sensor network, this is as if no sensor is assigned to some of the terminals and as a result the measurement rate from those terminals is $0$.  In both cases, one needs to characterize the required measurement rate region of the terminals for low-distortion recovery.

In this paper, we address a two terminal scenario for a memoryless distributed source $(X^n,Y^n)$. Memoryless property of the source implies that there is no temporal correlation between samples of the signals in each terminal. The spatial correlation between the sources is modeled by assuming that the samples of the signals $(X_i,Y_i)$ are generated by a probability distribution $p_{X,Y}$.  The extension to more than two terminals is also straightforward. 

This problem has been vastly studied under different signal structures and recovery algorithms (in particular \cite{DCS1,DCS2}) as an extension of the traditional single terminal compressed sensing introduced in \cite{CS1,CS2}. Specially, it has been attempted to make a connection between multi-terminal compressed sensing and the distributed source coding (Slepian-Wolf) counterpart in information theory (please refer to \cite{rice} for extra refrences). 

A  closely related work to our paper, is  the \textit{`analog to analog'} (A2A) compression problem first studied in \cite{verdu}, where it was proved that under some regularity conditions on encoder and decoder, the required measurement rate in order to recover the source with a negligible block error probability is given by the upper R\'enyi information dimension (RID) of the source. The results was extended to prove the noise stability of the decoder \cite{WV2}. In \cite{HAT1,HAT2}, it was proved that under a much weaker entropic distortion measure (compared with block error probability or MSE), a measurement rate of at least R\'enyi information dimension is still necessary to roughly capture the information of the source. Moreover, the polarization idea \cite{source_polarcode,channel_polarcode} was exploited to construct a family of deterministically truncated Hadamard matrices that universally capture the information of all probability distributions with a given RID.

In \cite{montanari}, using the spatially coupled  Gaussian matrices and the prediction of \cite{mezard1,mezard2}, it was rigorously proved that the measurement rate as large as the RID of the source is sufficient to stably recover the source with a negligible mean square error (MSE)  by running a feasible complexity approximate message passing algorithm (AMP) first developed in \cite{maleki} and rigorously analyzed in \cite{bayati}. 

In \cite{HAT2}, a characterization was given of the the measurement rate region of a memoryless multi-terminal signal in order to fully capture the information of the source in terms of the joint and conditional RID which was in spirit very similar to Slepian-Wolf (S\&W) region for  distributed source coding with the discrete entropy replaced by the RID and this region was also shown to be tight.

In this paper, we extend the results in \cite{montanari, HAT2} by developing a multi-terminal variant of AMP (MAMP) algorithm to study the multi-terminal compressed sensing of a well-behaved class of probability distributions (linearly correlated signals) for which the joint and the conditional RID's are well-defined. 
We use random Gaussian matrices (independent across different terminals) to take measurements and run the MAMP algorithm to reconstruct the source. We prove that the behavior of MAMP algorithm can be fully characterized by a two dimensional state\footnote{In general the number of states is equal to the number of terminals.} and at each iteration, this state changes according to an explicit state evolution (SE) equation. We use this SE equation to obtain the rate distortion curve of the source, where we use the mean square error (MSE) as the distortion measure. We also show that after spatially coupling of measurement matrices in each terminal, the low distortion measurement rate region can be fully characterized by the joint and the conditional RID's as predicted by \cite{HAT2}. 


\section{R\'enyi Information Dimension and Linearly Correlated Multi-terminal Sources}\label{RID_subsection}
Let $X$ be a scalar random variable with a probability distribution $p_X$ over $\mR$. The upper and lower RID of $X$ are defined by
\begin{align*}
\overline{d}(p_X)=\overline{d}(X)=\limsup_{q \to \infty} \frac{H([X]_q)}{\log_2(q)},\\
\underline{d}(p_X)=\underline{d}(X)=\liminf_{q \to \infty} \frac{H([X]_q)}{\log_2(q)},
\end{align*}
where for $x\in \mR$ and $q\in \mN$, $[x]_q=\frac{\lfloor q x \rfloor}{ q}$ denotes the quantization of $x$ by spacing $\frac{1}{q}$ and where  $\lfloor x \rfloor= \max \{k\in \mZ : k \leq x\}$.  If both limits coincide then we define $d(X)=\overline{d}(X)=\underline{d}(X)$. A parameter related to the RID is the MMSE dimension of $X$ defined in \cite{WV3}. Let 
\begin{align*}
\text{mmse}(s)=\E (X - \E (X|Y))^2, \ \  Y=\sqrt{s} X + Z,
\end{align*}
where $Z \sim {\normal(0,1)}$ is a  Gaussian random variable independent of $X$. The upper and lower MMSE dimension of $X$ are defined by
\begin{align*}
\overline{D}(p_X)=\overline{D}(X)=\limsup_{s \to \infty} s\ \text{mmse}(s)\\
\underbar{D}(p_X)=\underline{D}(X)=\liminf_{s \to \infty} s\ \text{mmse}(s),
\end{align*}
and if both limits coincide then we define $D(X)=\overline{D}(X)=\underline{D}(X)$. In \cite{WV3}, it was proved that if $H(\lfloor X \rfloor)<\infty$ then 
\begin{align*}
\underline{D}(X) \leq \underline{d}(X) \leq \overline{d}(X) \leq \overline{D}(X).
\end{align*}
Hence, if $D(X)$ exists so does $d(X)$ and they are equal. By Lebesgue decomposition theorem, any probability distribution like $p_X$ can be decomposed as a convex combination of continuous, discrete and singular parts, i.e. 
\begin{align*}
p_X=\alpha _c p_c + \alpha_d p_d + \alpha _s p_s,
\end{align*}
where $\alpha_c+\alpha_d+\alpha_s=1$ and $p_c,p_d$ and $p_s$ denote the continuous, discrete and the singular part of the distribution. R\'enyi proved that if $\alpha_s=0$, namely if $p_X$ has no singular part, then $d(p_X)$ is well defined and is equal to $\alpha_c$, the weight of the continuous part \cite{renyi}. Moreover, it was proved in \cite{WV3} that if $\alpha_s=0$ then $D(p_X)$ also exists and is equal to $d(X)=\alpha_c$.

For simplicity, we will restrict ourselves to the space of linearly correlated random variables introduced in \cite{HAT2}, where a $k$ dimensional random vector $S$ is linearly correlated if there is a sequence of independent non singular variables $Z^n$ and a $k\times n$ matrix $A$ such that $S= A Z^n$. This space is rich enough for most of the applications. Furthermore, over this space it is possible to give a full characterization of joint and conditional RID as in \cite{HAT2}. Appendix \ref{lin_cor} contains a brief overview of linearly correlated signals and how to compute their joint and conditional RID's.

\section{Statement of the Results}
\subsection{Gaussian Measurement Matrices}
Let $n \in \mN$ and let $(x^n,y^n)=\{(x_i,y_i)\}_{i=1}^n$ be a realization of a two terminal memoryless source $(X,Y)$ with a probability distribution $p_{X,Y}$. Let $\vecu =A x^n$ and $\vecv=B y^n$ be the measurement vectors, where $A$ is an $m_x \times n$ and $B$ is an $m_y \times n$ matrix whose components are  i.i.d. zero mean Gaussian random variables with variance $\frac{1}{m_x}$ and $\frac{1}{m_y}$ respectively. We define $\rho_x=\frac{m_x}{n}$ and $\rho_y=\frac{m_y}{n}$ as the measurement rates of  the two terminals. 

In order to recover the initial signal $(x^n,y^n)$, we propose the following joint message passing algorithm which is an extension of the single terminal message passing proposed in \cite{maleki}. We assign a variable node to each component of $x^n$ and $y^n$ and a check node to every measurement. Figure \ref{graphical} shows the resulting graphical model, where the internal check node between variable nodes $(x_i,y_i)$ show the correlation resulted because of the joint distribution $p_{X,Y}$. 

\begin{figure}[t!]
  \centering
  \includegraphics[width=0.5\textwidth] {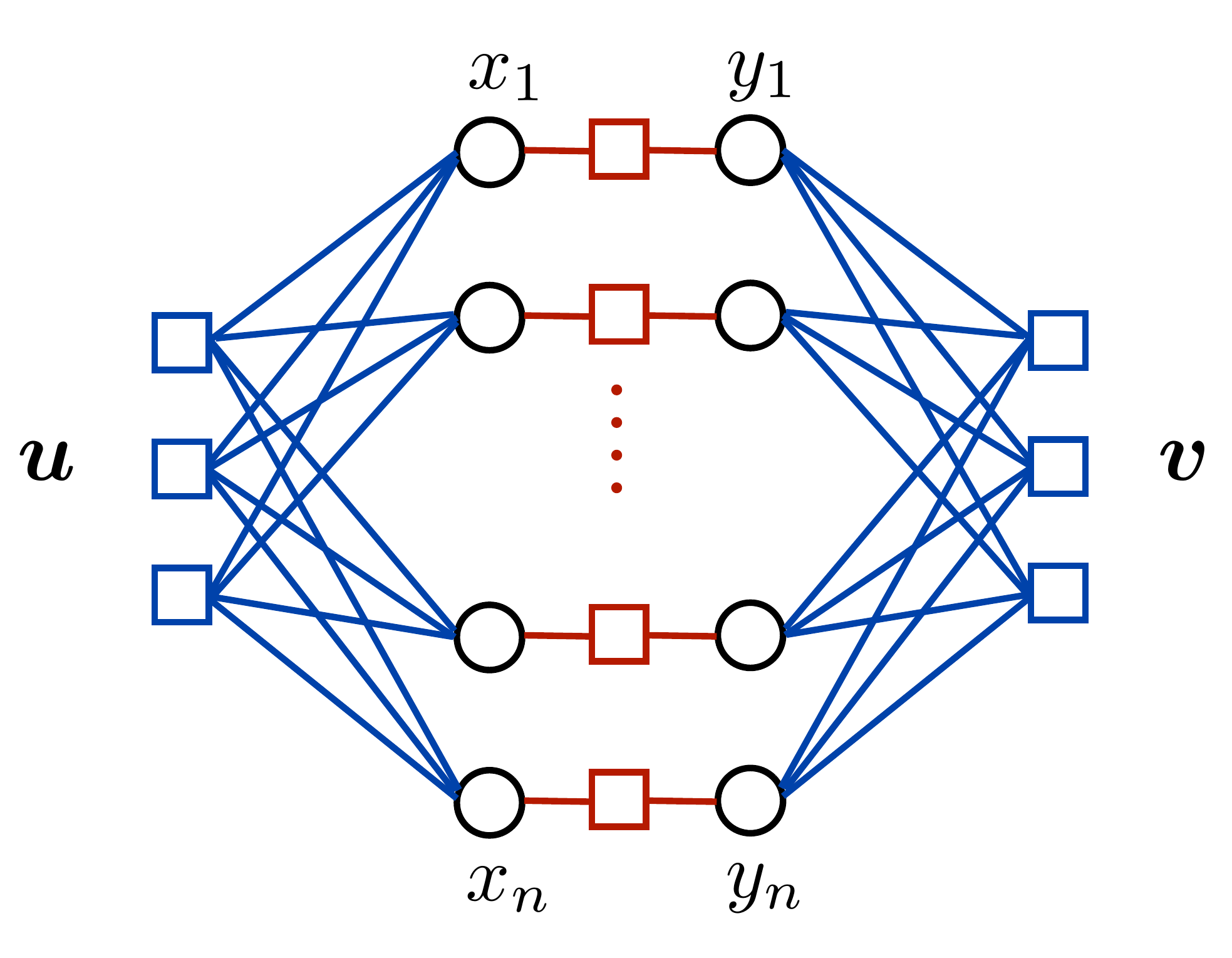}
    \vspace{-2mm}
  \caption{Graphical Model Representation for Two Terminal Compressed Sensing. The external check nodes correspond to  measurements whereas the internal check nodes between $x^n$ and $y^n$ represent the joint distribution $p_{X,Y}$ between $(X_i,Y_i)$.}
    \label{graphical}
\end{figure}

Let $a,b \in [m_x]$ and $i,j \in[n]$ be the indices for check and variable nodes in $T_X$ and let $c,d \in [m_y]$ and $k,l \in [n]$ denote the corresponding indices for $T_Y$. The multi-terminal message passing is given by
\begin{align}
r^t_{a \to i}& = u_a - \sum _{j \in [n]\backslash i} A_{a j} x^t_{j \to a},\label{MT_MP_first}\\
s^t_{c \to k}& = v_c - \sum _{l \in [n]\backslash k} B_{c l} y^t_{l \to c},\label{MT_MP_second}\\
x^{t+1}_{i\to a}&=\eta^x_t( \sum _{b \in [m_x] \backslash a} A_{b i} r^t_{b \to i}, \sum _{d \in [m_y]} B_{d i} s^t_{d \to i}),\label{MT_MP_third}\\
y^{t+1}_{k\to c}&=\eta^y_t(\sum _{b \in [m_x]} A_{b k} r^t_{b \to k}, \sum _{d \in [m_y]\backslash c} B_{k d} s^t_{d \to k})\label{MT_MP_last},
\end{align}
Notice that the only interaction between the messages in $T_X$ and $T_Y$ is via the threshold functions $\eta^x_t$ and $\eta^y_t$. In particular, if $\eta^x_t$ only depends on the first argument and if $\eta^y_t$ only depends on the second argument, this message passing algorithms is transformed to two independent message passing algorithms one running on $T_X$ and the other on $T_Y$.   
As the measurement matrices $A$ and $B$ are dense matrices with columns with $\ell_2$ norms close to $1$, it is possible to approximate the above message passing algorithm. This has been done heuristically in Appendix \ref{MT_AMP_Appendix}. The resulting MAMP (multi-terminal approximate message passing) algorithm is as follows initialized with $r^{-1}=0, s^{-1}=0$ and $x^{0}=y^{0}=0$:

\begin{align}
&r^t=\vecu- A x^t - \frac{ \bracket{\partial_1 \eta^x_t(A^* r^{t-1} + x^{t-1},B^*s^{t-1}+y^{t-1})}}{\rho_x} r^{t-1},\label{res_x}\\
&s^t=\vecv-B y^t - \frac{\bracket{ \partial_2 \eta^y_t(A^* r^{t-1} + x^{t-1},B^*s^{t-1}+y^{t-1})}}{\rho_x} s^{t-1},\label{res_y}\\
&x^{t+1}=\eta^x_t(A^* r^t + x^t,B^* s^t + y^t),\label{threshold_x}\\
&y^{t+1}=\eta^y_t(A^* r^t + x^t,B^* s^t + y^t),\label{threshold_y}
\end{align}
where $r^t \in \mR^{m_x}$ and $s^t \in \mR^{m_y}$ are the residual terms  and $x^t , y^t \in \mR^n$ are estimates of the signals at time $t$ and where for a function $f:\mR^2 \to \mR$, $\partial_1 f$ and $\partial_2 f$ denote the partial derivative of $f$ with respect to the first and the second argument respectively. Moreover, with some abuse of notation,  we assume that $\eta_t(g^l,h^l)=(\eta_t(g_1,h_1), \dots, \eta_t(g_l,h_l))$ applies component-wise. Also for an $n$ dimensional vector $u^n$, $\bracket{u^n}=\frac{1}{n}\sum_{i=1}^n u_i $ denotes the average of the elements of $u^n$.

It is also important to mention the appearance of Onsager terms in the Equations \eqref{res_x} and \eqref{res_y} as also mentioned in \cite{maleki,bayati}. This term can be considered as a second order correction for the mean field approximation of the message passing algorithm whose addition removes the correlation that exists between the fixed measurement matrices $A$ and $B$ and the estimated signal $(x^t,y^t)$ in the thermodynamic limit as the system size $n$ tends to infinity, which specially allows to completely describe the system state with a state evolution (SE) equation.

\begin{thm}\label{thm1}
Let $(x^n,y^n)$ be a realization of a memoryless source and assume that $(x^t, y^t)_{t\geq 0}$ is the output of the MAMP algorithm as in Equations \eqref{res_x}-\eqref{threshold_y} with Lipschitz continuous threshold functions $\eta^x_t$ and $\eta^y_t$. Let $\psi: \mR^2 \to \mR$ be a pseudo-Lipschitz function. Asymptotically as $n$ tends to infinity 
\begin{align*}
\frac{1}{n} \sum_{i=1}^n \psi(x_i, x^t_i)& \to \E\psi(X , \eta^x_t(X+\sqrt{\tau^t_x} Z_x , Y+ \sqrt{\tau^t_y} Z_y)),\\
\frac{1}{n} \sum_{i=1}^n \psi(y_i,y^t_i)&  \to \E\psi(Y , \eta^y_t(X+\sqrt{\tau^t_x} Z_x , Y+ \sqrt{\tau^t_y} Z_y))
\end{align*}
almost surely, where $(\tau_x^t,\tau_y^t)_{t\geq 0}$ satisfy the equation
\begin{align*}
\tau^{t+1}_x&=\sigma_x^2 + \frac{1}{\rho_x} \E(X - \eta^x_t(X+\sqrt{\tau^t_x} Z_x , Y+ \sqrt{\tau^t_y} Z_y))^2,\\
\tau^{t+1}_y&=\sigma_y^2 + \frac{1}{\rho_y} \E(Y - \eta^y_t(X+\sqrt{\tau^t_x} Z_x , Y+ \sqrt{\tau^t_y} Z_y))^2,
\end{align*}
with $\tau^{(-1)}_x=\tau^{(-1)}_y=\infty$, with $Z_x,Z_y$ zero mean unit variance Gaussian variables independent of each other and $X$ and $Y$ and with $\sigma_x^2$ and $\sigma_y^2$ denoting the measurement noise variance in $X$ and $Y$ terminals.
\end{thm}

\begin{proof}
Proof follows from the Bolthausen's conditioning technique used in \cite{bayati} with the only difference that one should apply the conditioning  to both terminals instead of a single terminal. 
\end{proof}

\begin{remark}
Theorem \ref{thm1} provides a single letter characterization of the asymptotic behavior of the MAMP, in the sense that to estimate a specific variable $(X_k,Y_k)$ the effect of all other variables is equivalent to adding a Gaussian noise with variance $(\tau_t^x,\tau^y_t)$. Moreover,  replacing $\psi(a,b)=(a-b)^2$ one gets the mean square error (MSE) of the estimator 
\begin{align*}
\frac{\|x^{t+1} - x\|_2^2}{n} \to \E(X - \eta^x_t(X+\sqrt{\tau^t_x} Z_x , Y+ \sqrt{\tau^t_y} Z_y))^2,\\
\frac{\|y^{t+1} - y\|_2^2}{n} \to \E(Y - \eta^y_t(X+\sqrt{\tau^t_x} Z_x , Y+ \sqrt{\tau^t_y} Z_y))^2.
\end{align*}
\end{remark}

We will also consider a noiseless case where $\sigma_x=\sigma_y=0$ which using the SE equation implies that the empirical  error after $t$ iteration is given by $\rho_x\tau_x^t$ and $\rho_y \tau_y^t$. One can also simply check that choosing $(\eta^x_t,\eta^y_t)$ to be the MMSE estimator minimizes the resulting error. We will always assume that the distribution of the signal is known and we will use the MMSE estimator for $(\eta^x_t,\eta^y_t)$, thus the resulting SE equation is 
\begin{align}
\tau^{t+1}_x&= \frac{1}{\rho_x} \mmse(X | X+\sqrt{\tau^t_x} Z_x , Y+ \sqrt{\tau^t_y} Z_y)\label{se_mmse1}\\
\tau^{t+1}_y&=\frac{1}{\rho_y} \mmse(Y | X+\sqrt{\tau^t_x} Z_x , Y+ \sqrt{\tau^t_y} Z_y)\label{se_mmse2}.
\end{align}

The behavior of MAMP depends on the stable set of the SE equation.  Proposition \ref{stable_fixed} states that for the special choice of MMSE estimators for $\eta^x_t$ and $\eta^y_t$, this stable set is a fixed point.
\begin{prop}\label{stable_fixed}
For a given $\rho_x,\rho_y$ and starting from $\tau_x^{(-1)}=\tau_y^{(-1)}=\infty$, the state vector $(\tau^t_x,\tau^t_y)$ given by SE equations in \eqref{se_mmse1}, \eqref{se_mmse2} converges to a well-defined fixed point.
\end{prop}
\begin{proof}
It is sufficient to prove that the resulting sequence is non-increasing thus converging to a well-defined fixed point. We use induction on $t$. For $t=0$, this  obviously holds because $\tau^0_x \leq \frac{\E(X^2)}{\rho_x} <\tau_x^{(-1)}=\infty$ and the same holds for $\tau_y^0$. Moreover, one can simply check that from the Data Processing inequality  $(\tau^{t+1}_x,\tau^{t+1}_y)$ are increasing function of $(\tau^{t}_x,\tau^{t}_y)$. Therefore, if from the induction hypothesis $\tau_x^t \leq \tau_x^{t-1}$ and $\tau_y^t \leq \tau_y^{t-1}$, it immediately results that $\tau_x^{t+1} \leq \tau_x^{t}$ and $\tau_y^{t+1} \leq \tau_y^{t}$.
\end{proof}

\subsection{Spatially Coupled Gaussian Measurement Matrices}
In the single terminal case, it has been already observed that with traditional Gaussian matrices, the the required measurement rate for complete recovery of the signal is far from the optimal rate given by the RID and spatial coupling is necessary to reduce the required measurement rate down to RID. The situation is very similar to coding theory where the BP threshold resulted from message passing algorithm is different from the optimal MAP threshold and extra spatial coupling is necessary to approach the optimal rate \cite{rudiger}.
%
 
 We briefly describe the structure of a spatially coupled measurement matrix as in \cite{montanari}. We consider a band diagonal weighting matrix $W$ of dimension $L_r\times L_c$ which is roughly row stochastic, i.e. $\frac{1}{2} \leq \sum_c W_{r,c} \leq 2$. In order to obtain the final measurement matrix we replace every entry $W_{r,c}$ by a i.i.d. $M \times N$ Gaussian matrix with entries having variance $\frac{W_{r,c}}{M}$, thus the final matrix will be $m \times n$ where $m=M L_r$ and $n=N L_c$ and the resulting measurement rate is $\rho=\frac{m}{n}=\frac{M L_r}{N L_c}$. 
 Figure \ref{band_diagonal}, borrowed from \cite{montanari}, shows a typical structure of a band diagonal matrix.
 
\begin{figure}[t!]
  \centering
  \includegraphics[width=0.5\textwidth] {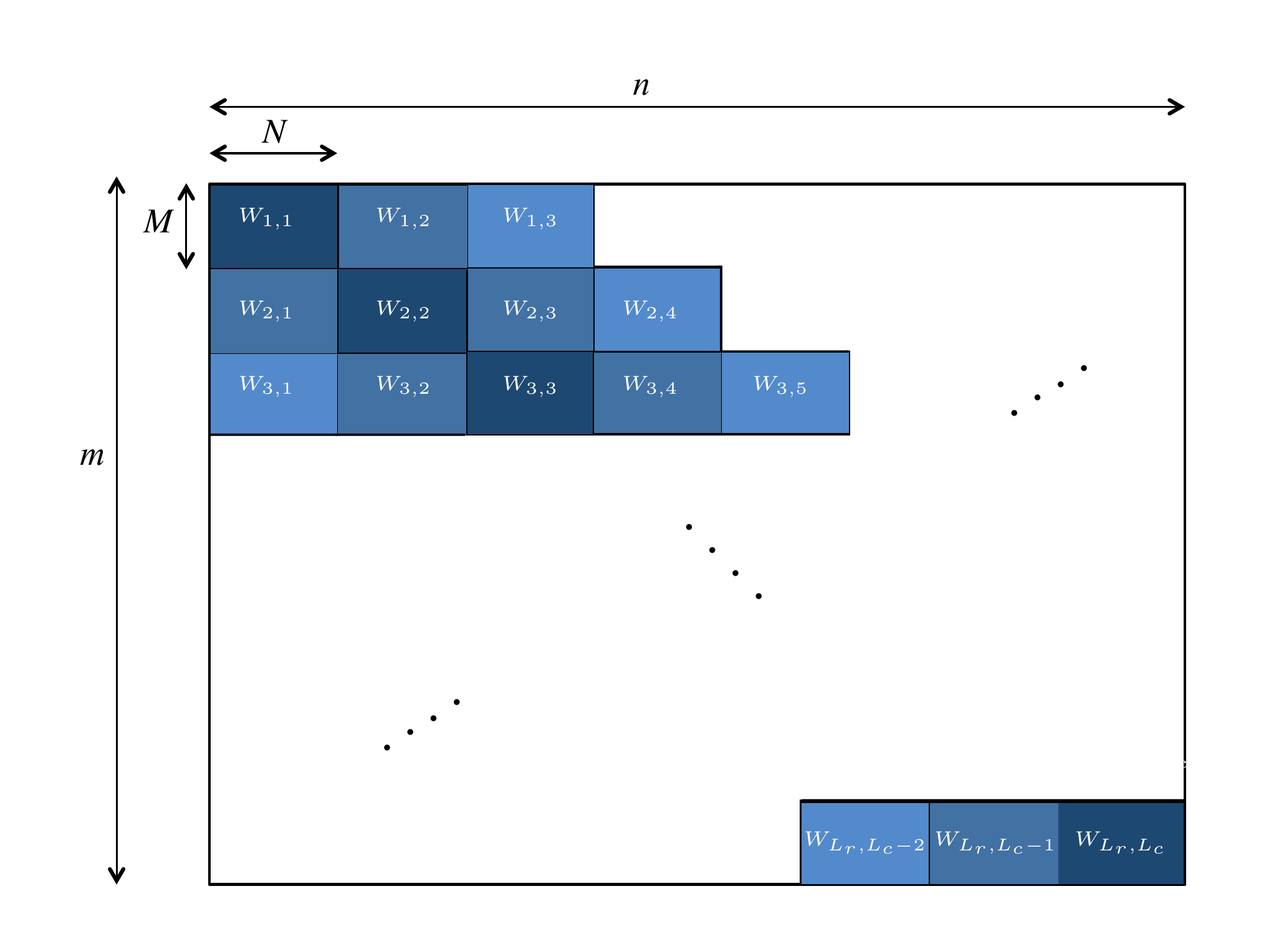}
    \vspace{-2mm}
  \caption{The Structure of A Band Diagonal Gaussian Matrix with Non-homogenous Entry Variances.}
    \label{band_diagonal}
\end{figure}

Each component of $W_{r,c}$ corresponds to one block containing an $M \times N$ matrix. Following the notations of \cite{montanari}, let $\col=\{1,2,\dots,L_c\}$ and $\row=\{1,2,\dots,L_r\}$ denote the row and column indices of these blocks. Let us define the following operators
\begin{align*}
\mmse_x(s_x,s_y)=\mmse(X|\sqrt{s_x}X+ Z_x, \sqrt{s_y}Y + Z_y),\\
\mmse_y(s_x,s_y)=\mmse(Y|\sqrt{s_x}X+ Z_x, \sqrt{s_y}Y + Z_y).
\end{align*}
In the two terminal case, for simplicity, we will use the same weight matrix in both terminals and the final measurement rate for each terminal can be controlled by the aspect ratio $\delta_x=\frac{M_x}{N_x}$ and $\delta_y=\frac{M_y}{N_y}$ of the corresponding sub matrices. 

\begin{definition}
For a roughly stochastic matrix of dimension $L_r\times L_c$, the state evolution sequence $\{\phi^x(t), \psi^x(t)\}_{t\geq 0}$ and $\{\phi^y(t), \psi^y(t)\}_{t\geq 0}$, $\phi^o(t)=(\phi^o_a(t))_{a \in \row}, \psi^o(t)=(\psi^o_i(t)) _{i \in \col}$  with $o \in \{x,y\}$ is defied as follows $\psi^o_i(0)=\infty, \  \ i \in \col$ 
and for all $t \geq 0$,
\begin{align}
\phi^o_a(t)&=\sigma_o^2 + \frac{1}{\delta_o} \sum_{i \in \col} W_{a,i} \psi^o_i(t),\label{sc_state1}\\
\psi^o_i(t+1)&=\mmse_o(\sum_{b \in \row} W_{b,i} \phi_b^x (t) ^{-1}, \sum_{b \in \row} W_{b,i} \phi_b^y (t) ^{-1}).\label{sc_state2}
\end{align}
where $\sigma_o^2$ is the variance of the measurement noise and $\delta_o=\frac{M_o}{N_o}$ is the measurement rate of the sub-matrices for terminal $o\in \{x,y\}$.
\end{definition}
Quantities $\psi_i(t)$ and $\phi_a(t)$ correspond to the asymptotic MSE of the MAMP. In particular, $\psi_i(t)$ is the asymptotic MSE of the variables located in block $i \in \col$ and $\phi_a(t)$ is the noise variance in the residual terms corresponding to row $a \in \row$ as we will explain later. 
Using $\{\phi, \psi\}$ sequence for each terminal it is possible to define the following MAMP algorithm. Let $Q^t$ be an $m \times n$ whose $i,j$ component is given by 
\begin{align}
Q^t_{i j}=\frac{\phi_r(t)^{-1}}{\sum _{k=1}^{L_r} W_{k c} \phi_k(t)^{-1}}\label{Qdef}
\end{align}
where $r$ is the row index of the measurement $i$ and $c$ is the column index of the variable $j$, thus it is a block constant matrix. We also define the MMSE threshold functions as 
\begin{align*}
\eta^x_{t,i}(g_i,h_i)= \E(X| X+s^x_i(t)^{-1} Z_x=g_i, Y+s^y_i(t) ^{-1} Z_y=h_i),
\end{align*}
with $s^o_i(t)=\sum_{u \in \row} W_{u,c} \phi^{o}_u(t)^{-1}$, where $c$ is the column index of variable $i$. MMSE estimator for $T_Y$ is defined similarly. We also assume that both of these estimator apply component wise, i.e. $\eta_{t}(g^l,h^l)=(\eta_{t,1}(g_1,h_1), \dots, \eta_{t,l}(g_l,h_l))$. With these notations the MAMP can be written as follows
\begin{align}
x^{t+1}&=\eta^x_t(x^t+ (Q_x^t \odot A)^* r_x^t, y^t +(Q_y^t \odot B)^*r_y^t),\label{sc_MAMP11}\\
r_x^t&=\vecu-A x^t + b_x^t \odot r_x^{t-1},\label{sc_MAMP12}\\
y^{t+1}&=\eta^y_t(x^t+ (Q_x^t \odot A)^* r_x^t, y^t +(Q_y^t \odot B)^*r_y^t),\label{sc_MAMP21}\\
r_y^t&=\vecv-B y^t + b_y^t \odot r_y^{t-1},\label{sc_MAMP22}
\end{align}
where $A$ and $B$ denote the spatially coupled measurement matrices, $\vecu=A x, \vecv=By$ are the measurements, $r_x$ and $r_y$ are residual terms, $Q_x$ and $Q_y$ are defined according to Equation \eqref{Qdef} and $b_x$ and $b_y$ are defined as follows.
Let $C(c)$ denote all the variables $i$ with column index $c \in \col$ and let
\begin{align*}
\bracket{\partial_1 \eta^x _t}_c=\bracket{ \partial _1 \eta^x_t(x_i^t+ ((Q_x^t \odot A)^* r_x^t)_i, y_i^t +((Q_y^t \odot B)^*r_y^t)_i}
\end{align*}
where the average is taken over all variables belonging to the the column block $c$. We define $b_x^t$ as a column vector of length $m$ which takes the same value for all components belonging to a row block $r\in \row$ and is defined as follows
\begin{align*}
b_{x,i}^t=\frac{1}{\delta_x} \sum_{c \in \col} W_{r_i, c} \tilde{Q}^{t-1}_{a_i,c} \bracket{\partial _1 \eta^x_{t-1}}_c,
\end{align*}
where $r_i$ is the row block that $i$ belongs to and $\tilde{Q}^{t}$ is a $L_r \times L_c$ matrix defined by $\tilde{Q}^t_{r,c}=Q^t_{x,\, ij}$ for any $i$ that belong to the row block $r$ and any column that belongs to the column block $c$. Notice that $Q^t$ itself is also block constant therefore it is not important which $i$ or $j$ is taken from the block.
 A similar expression holds for the $b_y^t$ by replacing $\partial _1 \eta_x^t$ by $\partial_2 \eta_y^t$, $Q_x^t$ by $Q_y^t$ and $\delta_x$ by $\delta_y$.

Using a similar steps as in \cite{montanari}, it is possible to show that the performance of the MAMP algorithm can be described by the state evolution given in Equation \eqref{sc_state1} and \eqref{sc_state2} where the number of states is equal to $2(L_r+L_c)$.

\begin{thm}
Let $(x^n,y^n)$ be a two terminal signal and let $\vecu=A x^n$ and $\vecv=B y^n$, where $A$ and $B$ are spatially coupled matrices with the same weight matrix $W$. Let $(x^t,y^t)$ be the output of MAMP algorithm in Equations \eqref{sc_MAMP11}-\eqref{sc_MAMP22}, where $\{\phi^o(t), \psi^o(t)\}_{t \geq 0, o \in \{x,y\}}$ is obtained from the SE equation \eqref{sc_state1}-\eqref{sc_state2}. Asymptotically, as $N_x,N_y$ go to infinity
\begin{align*}
\frac{1}{N_x} \sum_{j \in C(i)} (x_j - x^t_j)^2 \to \psi^x_i(t), \frac{1}{N_y} \sum_{j \in C(i)} (y_j - y^t_j)^2 \to \psi^y_i(t).
\end{align*}
\end{thm}

Based on the results proved in \cite{montanari} for the single terminal case and the lower bound provided in \cite{HAT2}, it is possible to give the following characterization for the achievable measurement rate region in the multi-terminal case.

\begin{thm}
Let $(X,Y)$ be a linearly correlated two terminal source and let $\rho_x, \rho_y \in [0,1]$ be such that
\begin{align}
\rho_x > d(X|Y),\, \rho_y> d(Y|X),\, \rho_x+\rho_y > d(X,Y).\label{sw_region}
\end{align}
There is an ensemble of spatially coupled measurement matrices that separately captures the signals in the two terminals and an MAMP algorithm that jointly recovers the signals in each terminal with a negligible  distortion.
\end{thm}
\begin{remark}
The optimal measurement rate region given by Equation \eqref{sw_region}, is very similar to the Slepian-Wolf rate region for distributed source coding where the RID in the compressed sensing setting plays a role similar to the discrete entropy in distributed source coding.
\end{remark}
\begin{proof}
We prove that the corner points $(d(X),d(Y|X))$ and $(d(Y),d(X|Y))$ are achievable under MAMP. In the single terminal case, If $\rho_x> d(X)$ asymptotically the signal in $T_X$ can be recovered with a negligible distortion. In multi-terminal case, if we consider only the terms related to $T_X$, from Equation \eqref{sc_state1}-\eqref{sc_state2}, we have 
\begin{align*}
\phi^x_a(t)&= \frac{1}{\rho_x} \sum_{i \in \col} W_{a,i} \psi^x_i(t),\\
\psi^x_i(t+1)&=\mmse_x(\sum_{b \in \row} W_{b,i} \phi_b^x (t) ^{-1}, \sum_{b \in \row} W_{b,i} \phi_b^y (t) ^{-1}).
\end{align*}
As $\eta_x, \eta_y$ are MMSE estimators, from Data Processing inequality, one can check that $\mmse_x(s_x,s_y)$ and $\mmse_y(s_x,s_y)$ are decreasing functions of $s_x$ and $s_y$. This implies that $\mmse_x(\sum_{b \in \row} W_{b,i} \phi_b^x (t) ^{-1}, \sum_{b \in \row} W_{b,i} \phi_b^y (t) ^{-1})$ is less than or equal to $\mmse_x(\sum_{b \in \row} W_{b,i} \phi_b^x (t) ^{-1}, 0)$, which is equal to the variance of the MMSE estimator for $X$ which does not use the information of $Y$. One can also check that SE equation is increasing with respect to $\psi_i^x(t)$, which implies that the $\psi^x$ sequence for the MAMP is dominated by the $\psi^x$ sequence of a single terminal AMP, which converges to $0$ for any $\rho_x > d(X)$. If $\psi^x_i(t)$ converges to zero so does the $\phi^x_a$ sequence, thus the SE equation for $T_Y$ will be as follows

\begin{align*}
\phi^y_a(t)&= \frac{1}{\rho_y} \sum_{i \in \col} W_{a,i} \psi^y_i(t),\\
\psi^y_i(t+1)&=\mmse_y(\infty, \sum_{b \in \row} W_{b,i} \phi_b^y (t) ^{-1}),
\end{align*}
which using the same steps as in the single terminal case, can be proved to converge to zero provides that 
\begin{align*}
\rho_y > \limsup_{s \to \infty} s \mmse(Y|X,\sqrt{s}Y + Z_y)=d(Y|X),
\end{align*}
where $Z_y$ is a zero mean unit variance Gaussian noise and where we used the fact that for the class of linearly correlated signals that we use, $d(Y|X)$ is well defined. 

Similarly, it is possible to prove that $(\rho_x,\rho_y)=(d(X|Y),d(Y))$ is also achievable. Furthermore, any point on the dominant face is also achievable because if we consider two ensembles of measurement matrices $(A_1,B_1)$ and $(A_2,B_2)$ with rate vectors $\vec{R}_1=(d(X),d(Y|X))$ and $\vec{R}_2=(d(Y),d(X|Y))$ achieving the two corner points respectively, by diagonally concatenating $r$ copies of the former with $s$ copies of the latter, one can get an ensemble with measurement rate $\frac{r}{r+s} \vec{R}_1 + \frac{s}{r+s}\vec{R}_2$ and a negligible distortion. 

The other points on the region are also achieved because their measurement rate is larger than or equal to the measurement rate of at least one point on the dominant face, thus their distortion will be asymptotically negligible as well. 
\end{proof}

\section{Simulation Results}\label{sim}
\subsection{Signal Model}
For simulation, we will use a linearly correlated random vector from $\lin_2$ whose independent constituents are random variables with Bernoulli-Gaussian distribution. Let $Z^k$ be a sequence of independent random variables with probability distribution $p_i(z)=(1-\alpha_i) \delta_0(z) + \alpha_i \normal(0, \frac{1}{\alpha_i})$ where $\delta_0$ is a delta measure at point zero and $\normal(0,\sigma^2,z)$ denotes a zero mean Gaussian distribution with variance $\sigma^2$. One can simply check that $\var (Z_i)=1$ and $d(Z_i)=\alpha_i$. Let $\Phi$ be a $2 \times k$ real-valued matrix. The two terminal linearly correlated source is given by $\Phi Z^k$. As explained in Section \ref{RID_subsection}, the joint and conditional RID of this source is well-defined. Notice that depending on the values of $\alpha_i$ and the structure of the matrix $\Phi$, this model can cover a wide variety of correlations between the signals in two terminals. In Appendix \ref{mmse_bg}, we have obtained a closed form expression for the MMSE estimator $(\eta^x, \eta^y)$ of this source in presence of the Gaussian measurement noise which we will use as a denoising (threshold) function in MAMP algorithm.

\subsection{Performance without Spatial Coupling}
In this section, we use the message passing algorithm given by Equations \eqref{res_x}-\eqref{threshold_y} to recover a linearly correlated Bernoulli-Gaussian signal for the noiseless case where there is no measurement noise.

\subsubsection{{\bf Comparison of the Empirical Results and SE predictions}}\label{comparison}
We consider a very simple case where $Z_1,Z_2,Z_3$ are three Bernoulli-Gaussian random variables with $d(Z_1)=d(Z_3)=0.2$ and $d(Z_2)=0.3$. The signal for the two terminals is given by $X=Z_1+Z_2$ and $Y=Z_2+Z_3$, thus $Z_1$ and $Z_3$ are the private parts of the signals and $Z_2$ is the common part which creates correlation between $X$ and $Y$. It is easy to check that $d(X)=d(Y)=0.44$ and $d(X|Y)=d(Y|X)=0.248$. 

Figure \ref{figx_05_06}, \ref{figy_05_06} show the simulation results for $\rho_1=0.5, \rho_2=0.6$. It is seen that there is a good match between the empirical variance of the estimator and the predictions of the SE. Moreover, the algorithm can not fully recover the signal which means that the SE equation has a fixed point other than $(\tau_x,\tau_y)=(0,0)$. The simulations has been repeated in Figure \ref{figx_05_07}, \ref{figy_05_07} by increasing the measurement rate of the $T_Y$ from $\rho_2=0.6$ to $0.7$. Plots show that this time MAMP algorithm successfully recovers the signal of both terminals. It is also important to notice that because of the  correlation between the terminals, increasing $\rho_2$ is helpful for recovering the signal in $T_Y$.
\begin{figure}[t!]
  \centering
  \includegraphics[width=0.5\textwidth] {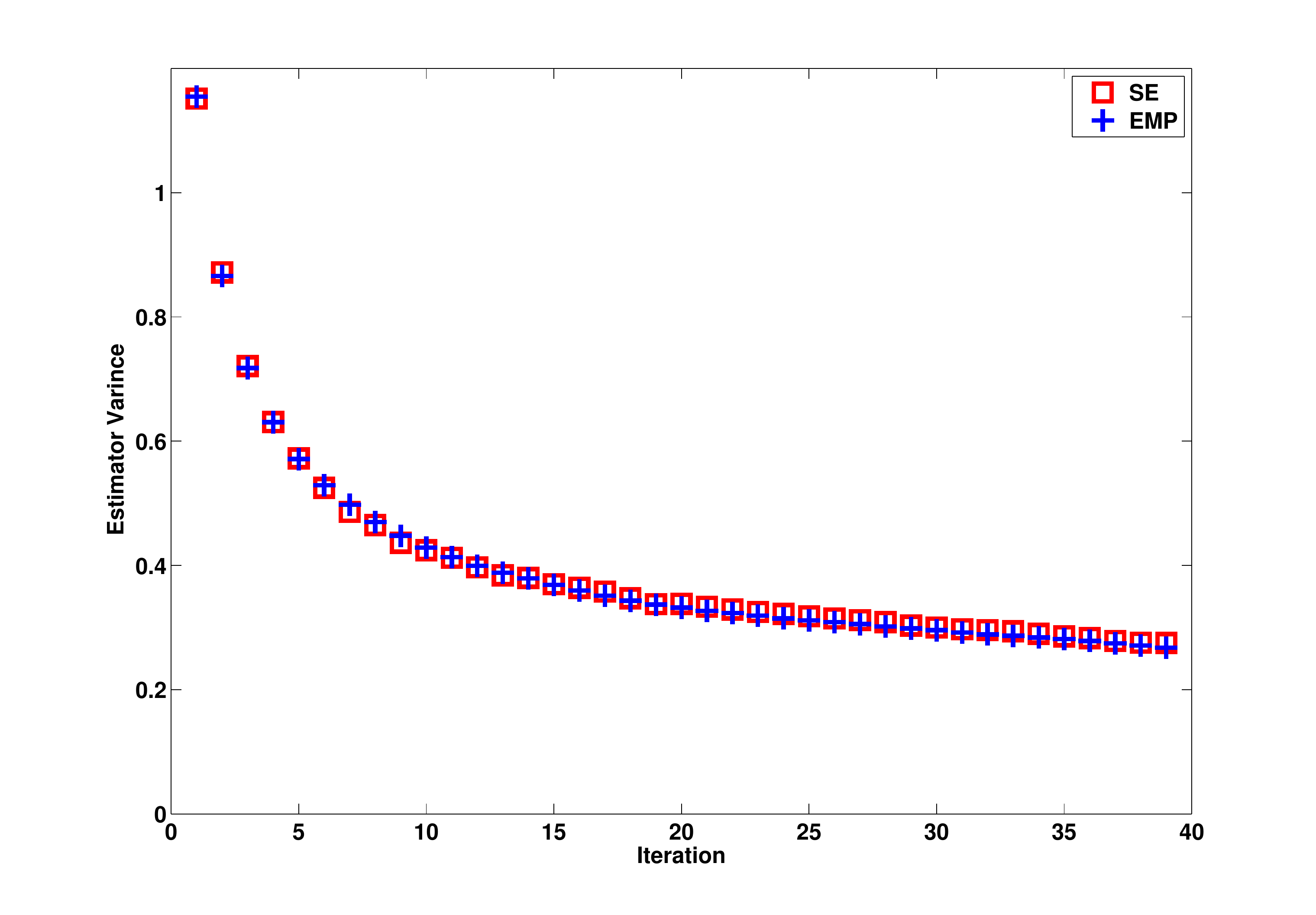}
  \vspace{-2mm}
  \caption{Empirical and SE Result for $T_X$ for $\rho_1=0.5, \rho_2=0.6$}
  \label{figx_05_06}
\end{figure}

\begin{figure}[t!]
  \centering
  \includegraphics[width=0.5\textwidth] {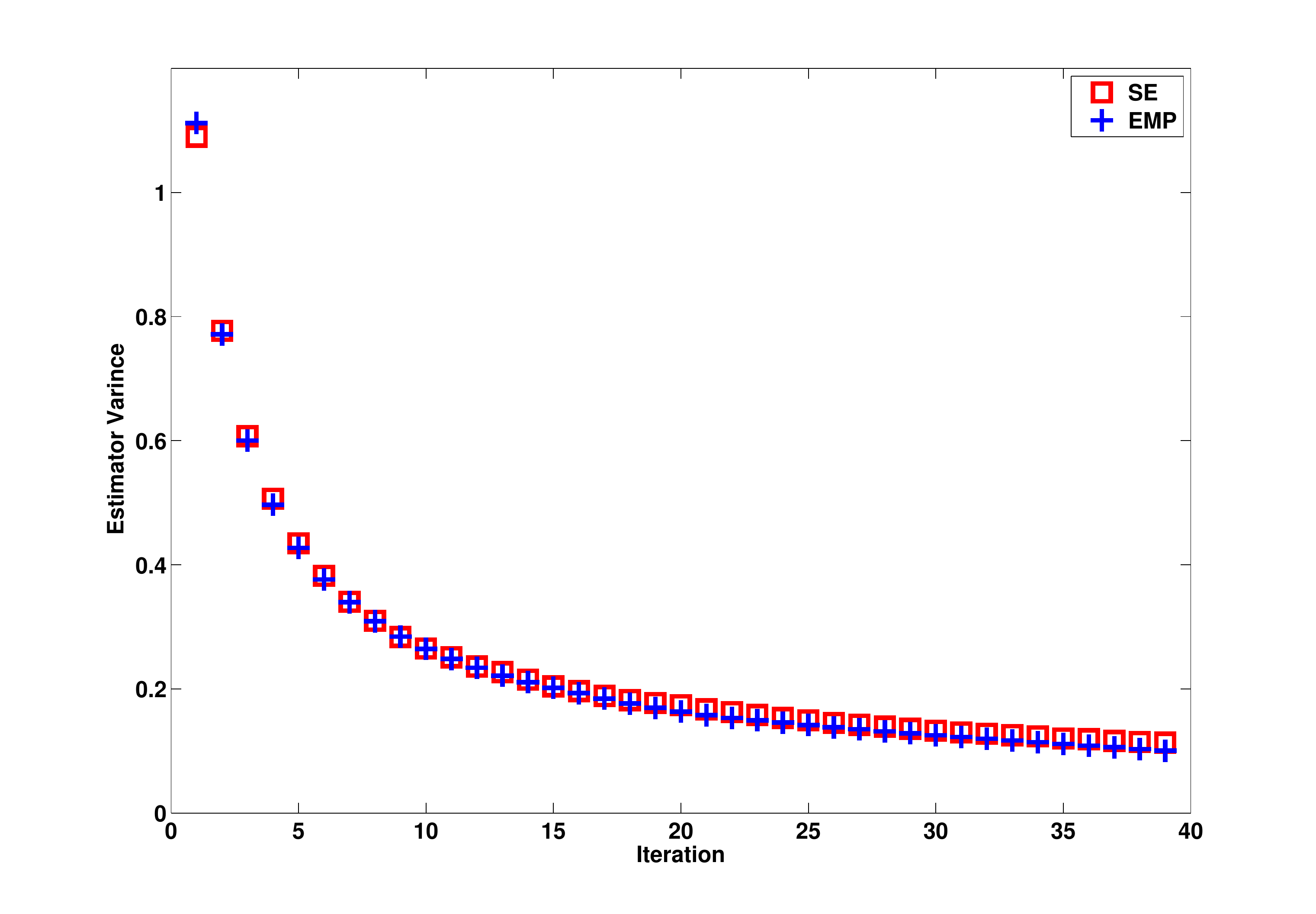}
    \vspace{-2mm}
  \caption{Empirical and SE Result for $T_Y$ for $\rho_1=0.5, \rho_2=0.6$}
    \label{figy_05_06}
\end{figure}

\begin{figure}[t!]
  \centering
  \includegraphics[width=0.5\textwidth] {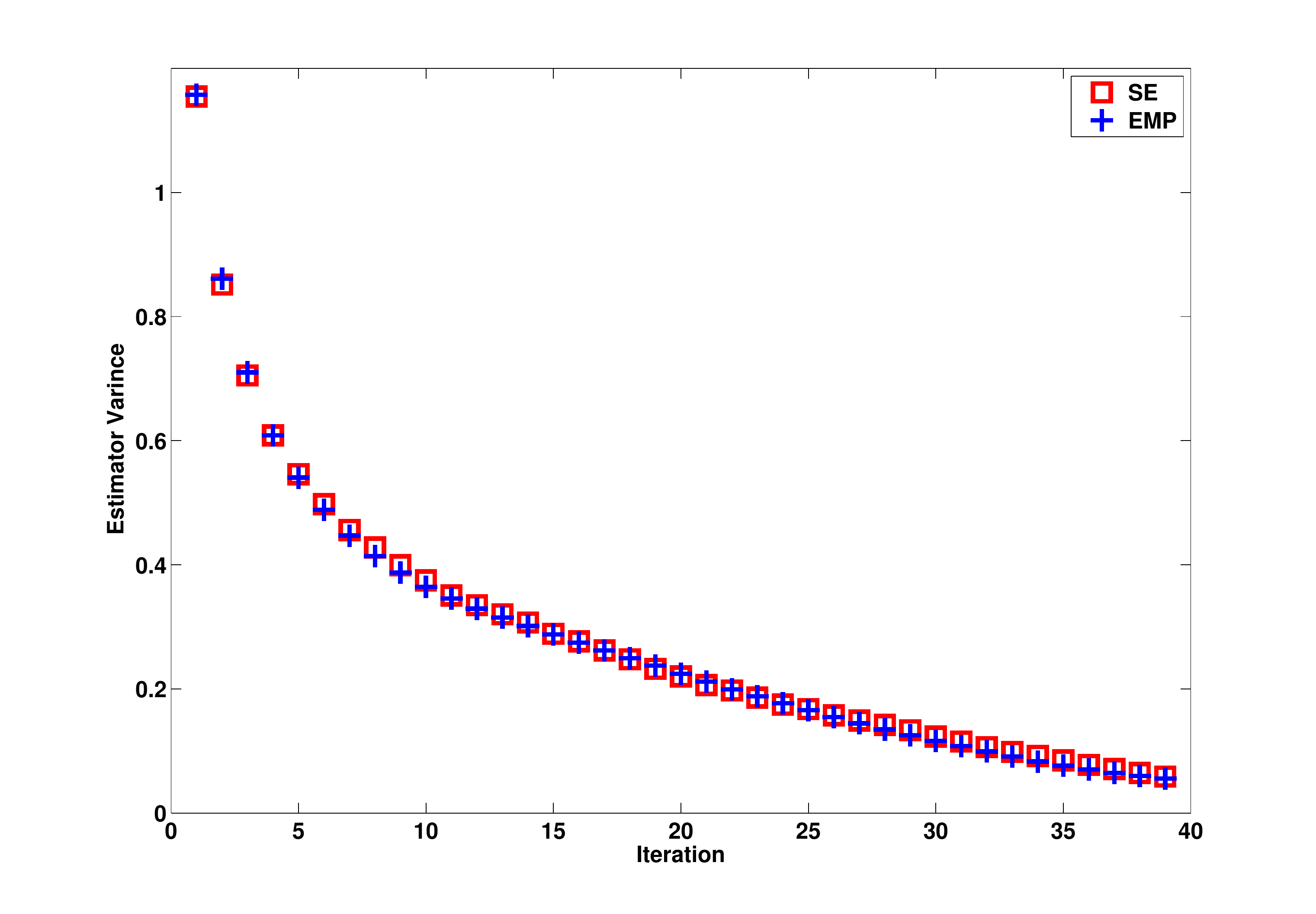}
    \vspace{-2mm}
  \caption{Empirical and SE Result for $T_X$ for $\rho_1=0.5, \rho_2=0.7$}
    \label{figx_05_07}
\end{figure}

\begin{figure}[t!]
  \centering
  \includegraphics[width=0.5\textwidth] {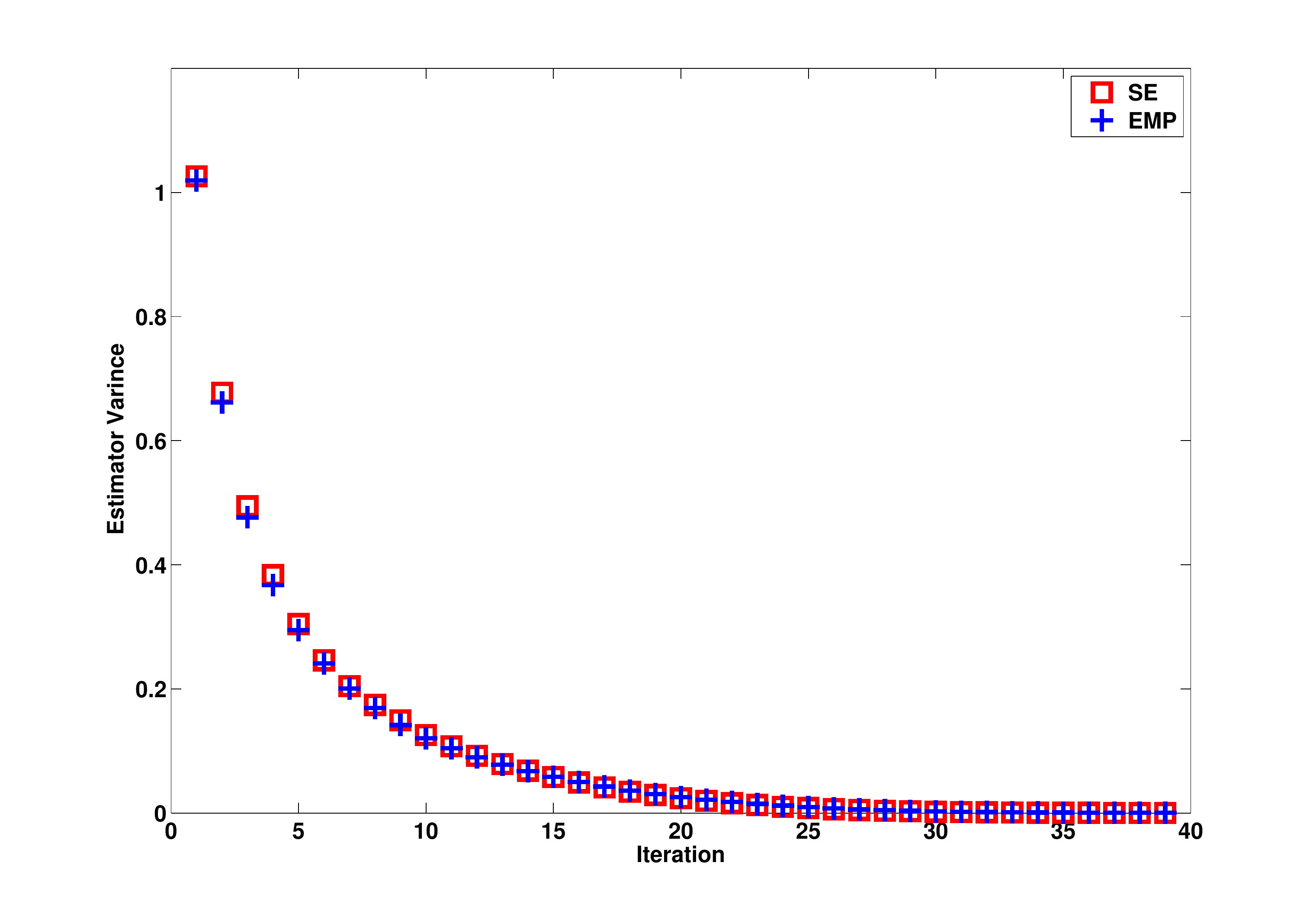}
    \vspace{-2mm}
  \caption{Empirical and SE Result for $T_Y$ for $\rho_1=0.5, \rho_2=0.7$}
    \label{figy_05_07}
\end{figure}

\subsubsection{{\bf Rate-Distortion Region}}
In this part, we run the MAMP algorithm for the same signal as in Section \ref{comparison} for different measurement rates. As a distortion measure, we consider the average of the mean square error of the two terminals. Figure \ref{rate_dist1},  \ref{rate_dist2}, \ref{rate_dist3} show a contour plot of the Rate-Distortion curve for three sources with the same individual but different conditional RID. The dashed lines show the boundary of the optimal pentagon. Low distortion recovery is not possible outside of this region. 

In the extreme case where the signals in two terminals are independent from each other, i.e. there is no common signal, the pentagon region reduces to a square region. On the contrary, if there is no private signal then the signals in both terminals are the same and the problem is reduced to a simple single terminal problem.  Obviously in this case, beacause of the independence of measurement matrices in the two terminals, individual measurement rates are not important as far as their sum is large enough. This can be seen from Figure \ref{rate_dist1},  \ref{rate_dist2}, \ref{rate_dist3} where we keep $d(X)=d(Y)=0.44$ but gradually increase the share of the common signal where as a result $d(X|Y)$ and $d(Y|X)$ start to decrease. It is observed that the contour lines gradually become parallel with $\rho_x+\rho_y=\text{constant}$.  

Notice that there is a huge gap between the low-distortion curve (distortion equal to $0.1$) and the optimal region. As we will see this gap is filled by using spatial coupling and running MAMP.

\begin{figure}[t!]
  \centering
  \includegraphics[width=0.5\textwidth] {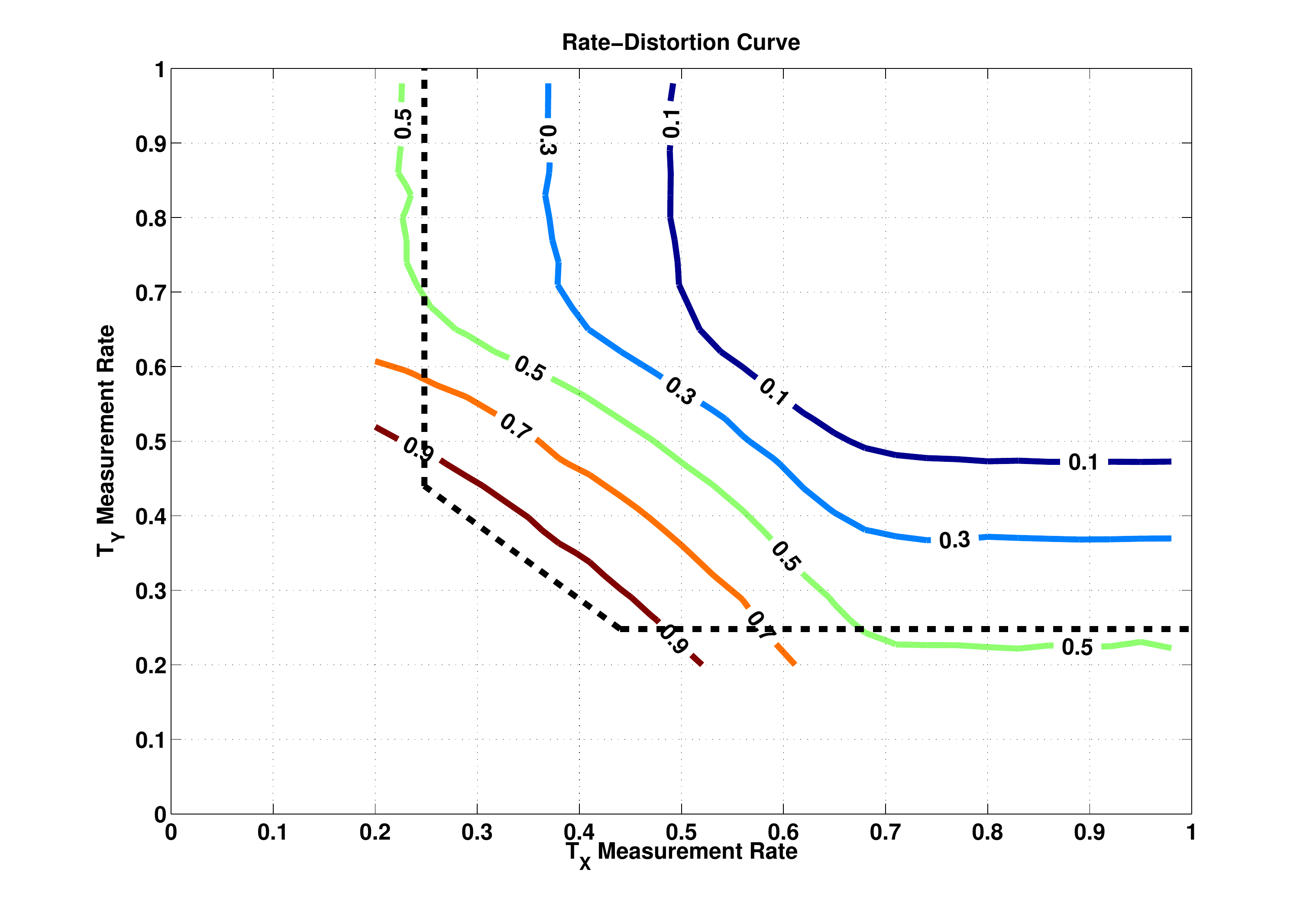}
    \vspace{-2mm}
  \caption{Rate-Distortion region for a linearly correlated Bernoulli-Gaussian source with $d(X)=d(Y)=0.44$ and $d(X|Y)=d(Y|X)=0.248$. The dashed lines show the boundaries of the optimal region.}
    \label{rate_dist1}
\end{figure}

\begin{figure}[t!]
  \centering
  \includegraphics[width=0.5\textwidth] {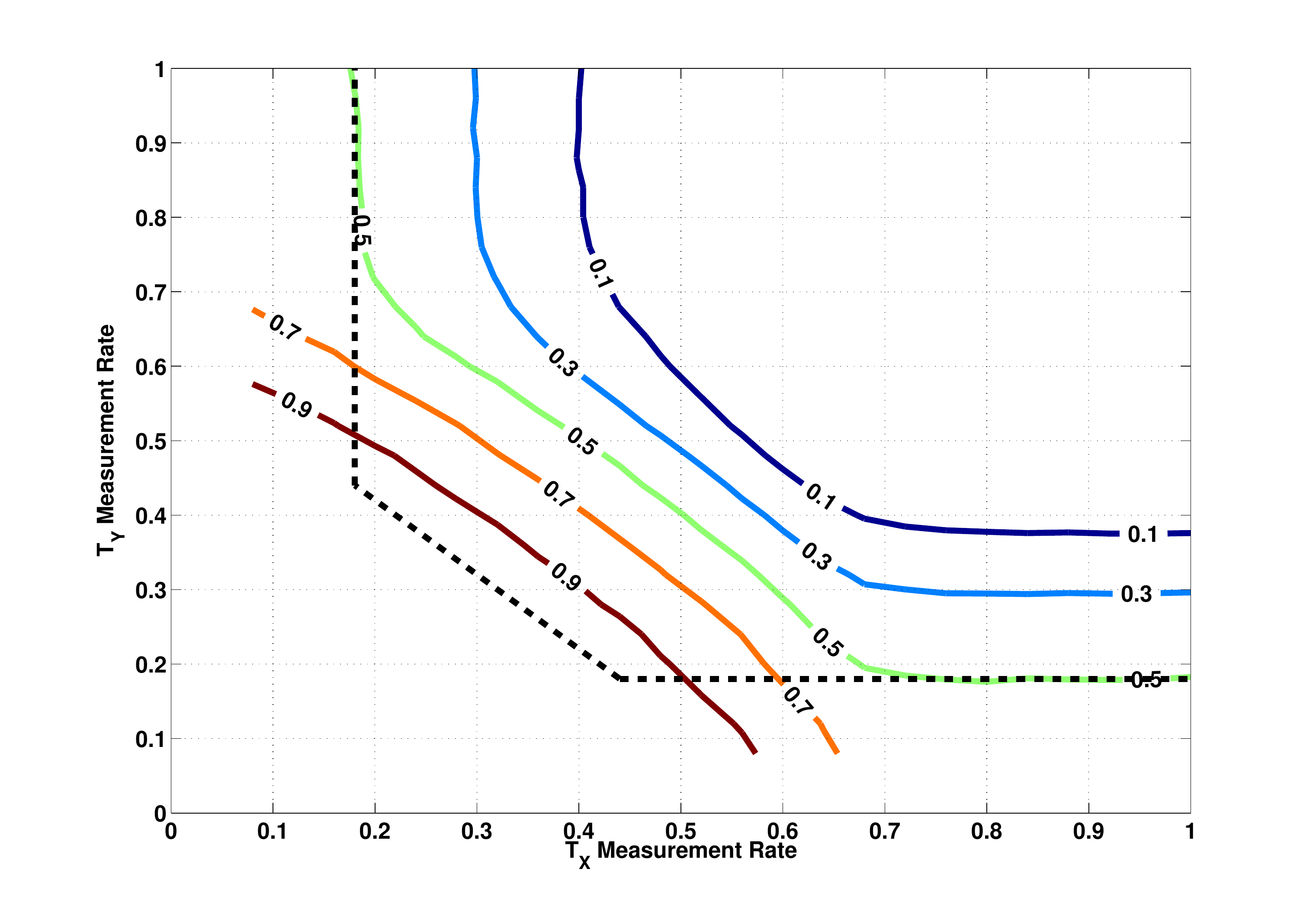}
    \vspace{-2mm}
  \caption{Rate-Distortion Region for the Linearly Correlated Bernoulli-Gaussian Source with $d(X)=d(Y)=0.44$ and $d(X|Y)=d(Y|X)=0.1802$.}
    \label{rate_dist2}
\end{figure}

\begin{figure}[t!]
  \centering
  \includegraphics[width=0.5\textwidth] {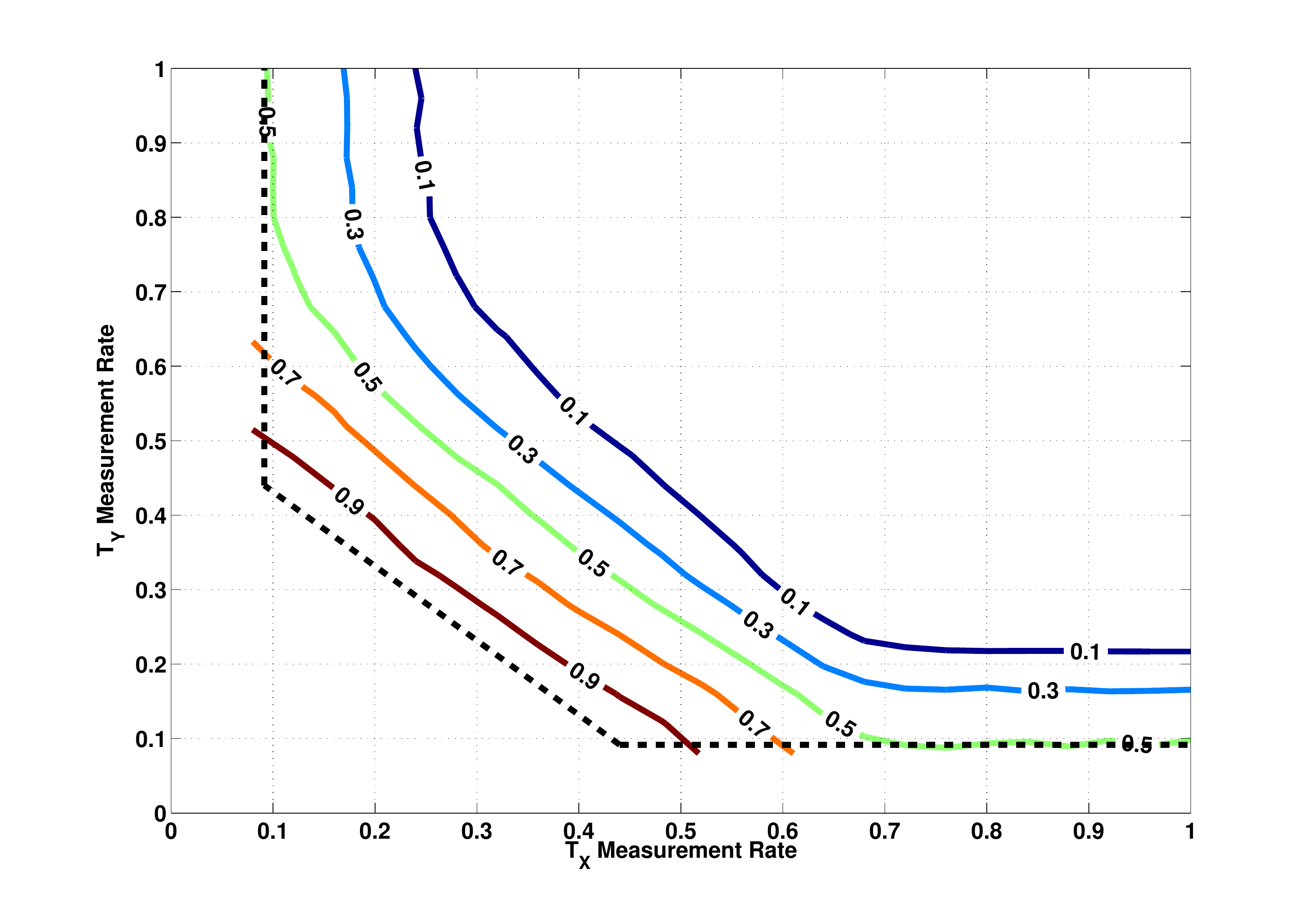}
    \vspace{-2mm}
  \caption{Rate-Distortion region for a linearly correlated Bernoulli-Gaussian source with $d(X)=d(Y)=0.44$ and $d(X|Y)=d(Y|X)=0.0916$. The dashed lines show the boundary of the optimal region.}
    \label{rate_dist3}
\end{figure}

\subsubsection{{\bf Effect of Correlation between the Terminals}}
In order to investigate the effect of correlation between the two terminals, we have plotted a low distortion contour of the three sources with the same $d(X)=d(Y)=0.44$ but three different conditional RID $0.248$, $0.1820$ and $0.0916$. Decreasing the conditional RID while fixing the individual entropy, make the signals in two terminals more correlated. A low distortion curve of the three sources is plotted in Figure \ref{increasing_correlation}. The plot shows that the required measurement rate is decreasing by increasing the correlation.
\begin{figure}[t!]
  \centering
  \includegraphics[width=0.5\textwidth] {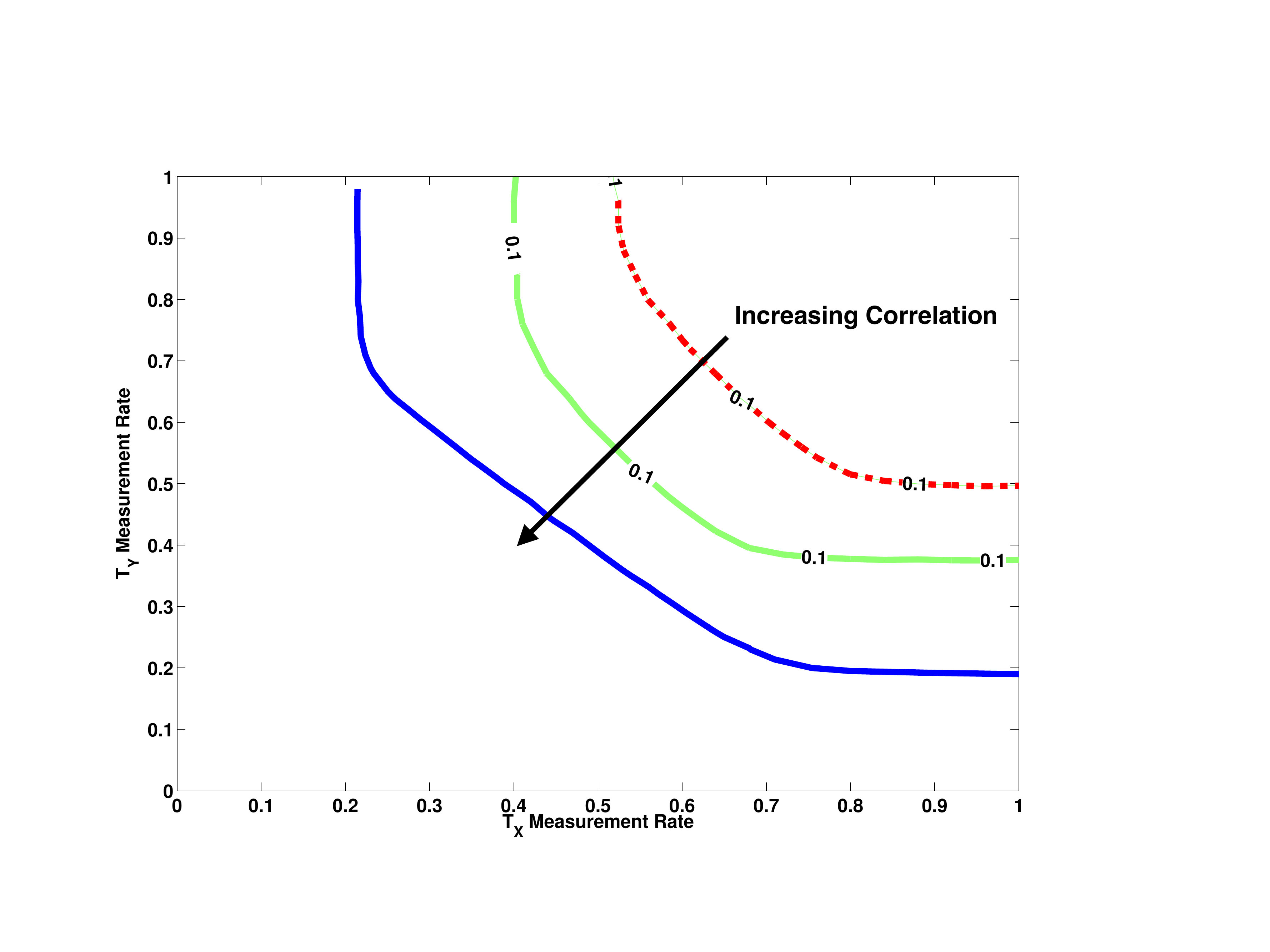}
    \vspace{-2mm}
  \caption{Effect of Correlation on Measurement Rate Region. The low distortion curve of three different two terminal sources with the same individual RID is plotted. The required measurement region of the more correlated source is dominated by that of the less correlated one.}
    \label{increasing_correlation}
\end{figure}

\subsection{Performance with Spatial Coupling}
In this section, we simulate the SE equation for MAMP algorithm. We consider the same source as in Section \ref{comparison} where $d(X)=d(Y)=0.44$ and $d(X|Y)=d(Y|X)=0.248$. In order to approach the corner point $(d(X),d(Y|X))$, we consider a measurement rate with $10$ percent oversampling, i.e. $\rho_x=1.1 d(X)$ and $\rho_y=1.1 d(Y|X)$. The simulation results has been shown in Figure \ref{SC_wave1}. Similar to the single terminal case, one can observe a wave-like phenomenon which starts from the boundary variables and proceeds towards the center recovering the variables gradually. In particular, to create the initial wave at the boundary one needs to oversample the boundary variables. Figure \ref{SC_wave4} depicts the simulation results for another experiment where $\rho_x$ is kept fixed but $\rho_y$ is reduced. It is observed that, this time spatial coupling wave proceeds to decode the variables in $T_X$ however the initially generated wave in $T_Y$ stops after a while and can not proceed to recover the all the variables in $T_Y$. 

By checking the results for non-spatially coupled case, one can see that the resulting MSE error decreases gradually by increasing the measurement rate.  On the contrary, in the spatially coupled case, either wave proceeds and recovers all the variables or it stops, thus asymptotically, there is a sharp transition in the resulting MSE in terms of measurement rate. 

For the same source, we have done the simulations to find boundary of the phase transition. Figure \ref{SE_MAMP} depicts the simulation result.

\begin{figure}[t]
  \centering
  \includegraphics[width=0.5\textwidth] {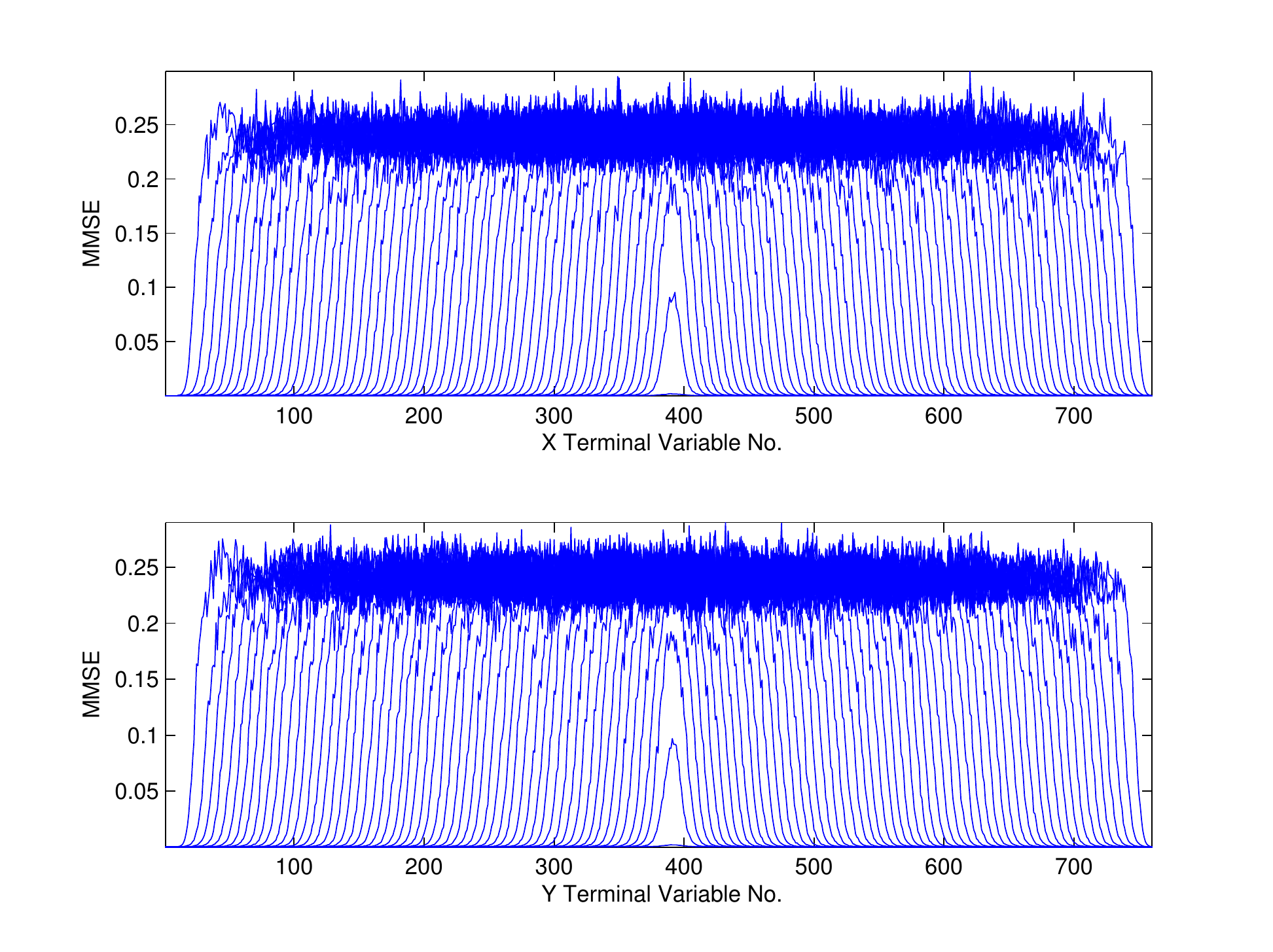}
    \vspace{-2mm}
  \caption{Spatial Coupling Wave for A Linearly Correlated Source with $\rho_x=1.1 d(X)$ and $\rho_y=1.1 d(Y|X)$}
    \label{SC_wave1}
\end{figure}

\begin{figure}[t]
  \centering
  \includegraphics[width=0.5\textwidth] {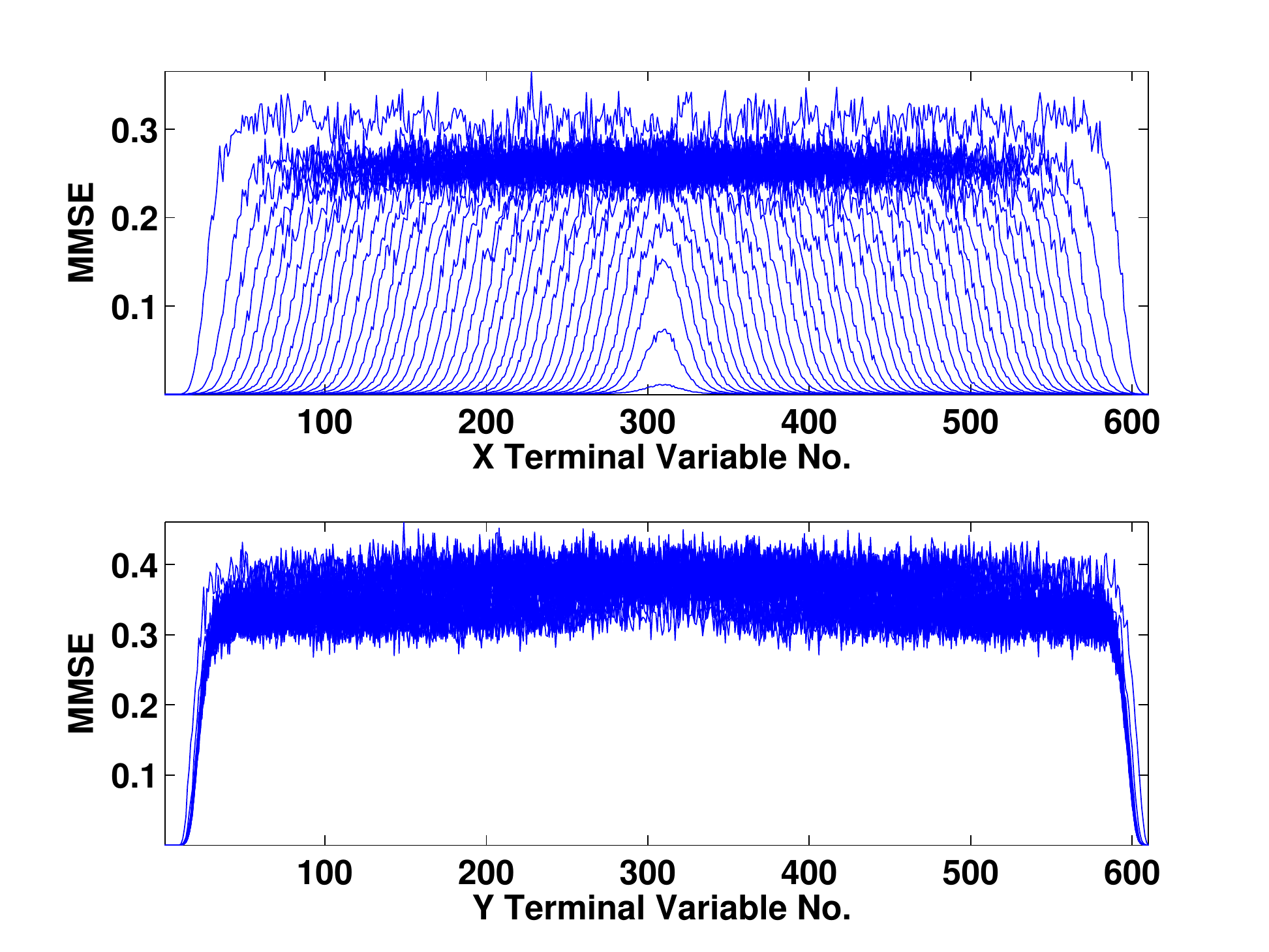}
    \vspace{-2mm}
  \caption{Spatial Coupling Wave for A Linearly Correlated Source with $\rho_x=1.1 d(X)$ and $\rho_y< d(Y|X)$}
    \label{SC_wave4}
\end{figure}

\begin{figure}[t]
  \centering
  \includegraphics[width=0.5\textwidth] {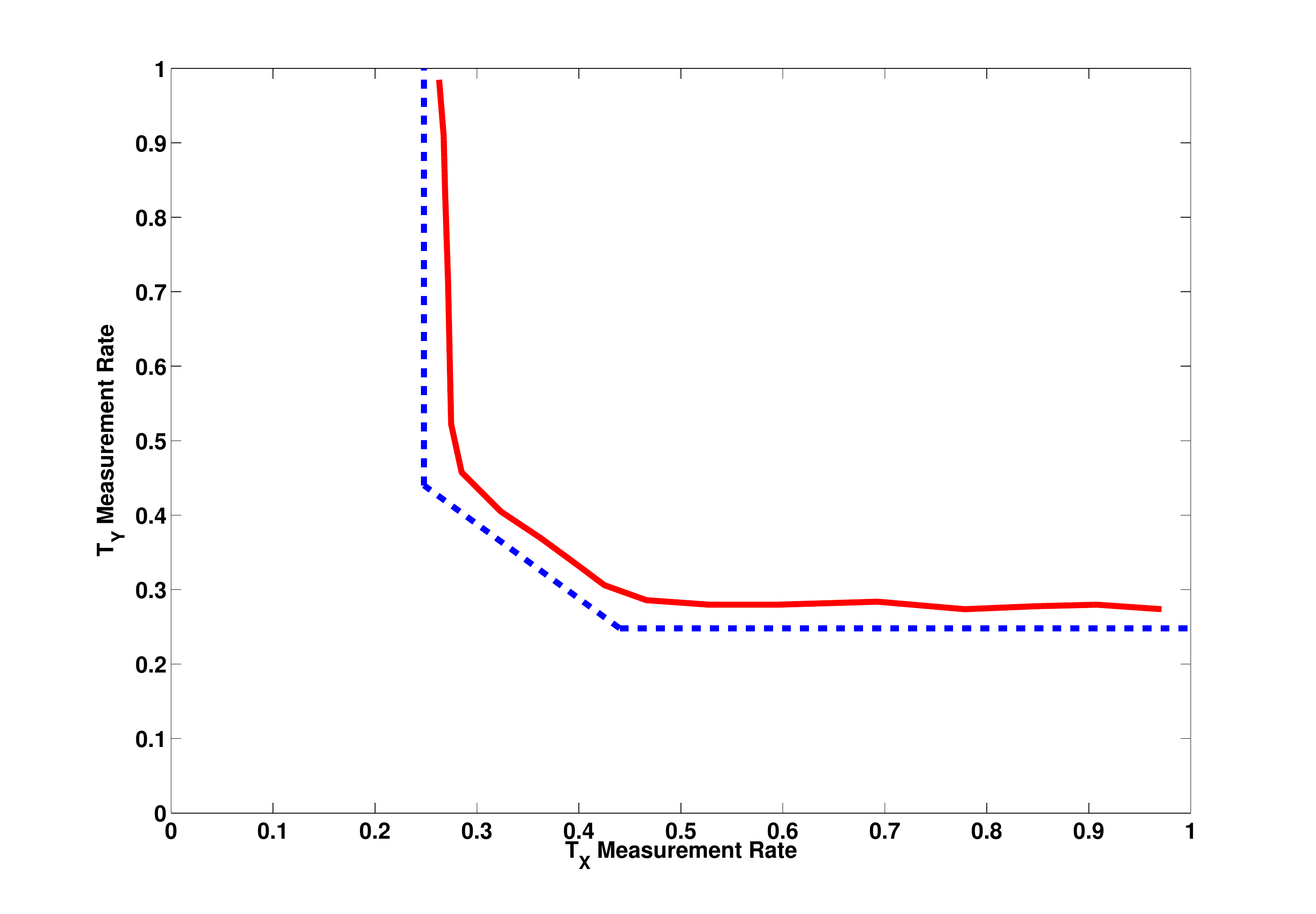}
    \vspace{-2mm}
  \caption{Phase Transition Boundary for MAMP and Comparison with SE Prediction. Dashed curve shows the theoretical boundary of the achievable measurement rate region.}
    \label{SE_MAMP}
\end{figure}

\bibliographystyle{IEEEtran} 
\bibliography{bibliography}

\appendices

\section{Linearly Correlated Random Signals and RID}\label{lin_cor}
Let $\lino$ be a probability space. We define the space $\lin _1$ as the space of all nonsingular scalar random variables. For $k\in \mN, k\geq 2$, $\lin _k$ is the space of all $k$-dimensional random vectors that can be written as a linear combination of finitely many independent nonsingular random variables, i.e. for $X^k \in \lin _k$ there is a $k \times n$ matrix $A$ and $n$ independent nonsingular random variables $Z^n$ such that $X^k=A Z^n$. For this set of random variables, the RID as proved in \cite{HAT2}, is well defined and can be obtained by the following formula 
\begin{align*}
d(X^k)=\E \{ \rank (A[C_\Theta])\},
\end{align*}
where $\{\Theta_i\}_{i=1}^n$ are independent binary random variables with $\pp(\Theta_i=1)=d(Z_i)$, $C_\Theta=\{i\in [n]: \Theta _i=1\}$ is a random subset of $[n]$ denoting the position of nonzero values in $\Theta^n$ and where for a subset $S \in [n]$, $A[S]$ is the matrix consisting of only those columns of $A$ with index in $S$. For example let, $d(Z_1)=d(Z_3)=0.2$ and $d(Z_2)=0.3$, $A=\left ( \begin{matrix} 1 & 1 & 0 \\ 0 & 1 & 1 \end{matrix} \right )$ and $X^2=A Z^3$. In this case one can simply check that $\rank(A[C_\Theta])=\Theta_1+\Theta_2+\Theta_3-\Theta_1 \Theta_2 \Theta_3$ because adding selecting any column adds $1$ unit to the rank unless all of the columns have been selected. Taking the expected value, one gets $d(X^2)=0.7-0.012=0.688$. Similarly, one can show that $d(X_1)=d(X_2)=0.44$. Thus, $d(X_1|X_2)=d(X_2|X_1)=d(X^2)-d(X_1)=0.248$.

The space of linearly correlated random variables is defined as $\lin=\cup _{i=1}^\infty \lin_i$. If $(X^k,Y^l)$ is a $k+l$ dimensional vector in  $\lin _{k+l}$, the conditional RID of $X^k$ given $Y^l$ is defined as follows
\begin{align*}
d(X^k|Y^l)=d(X^k,Y^l)-d(Y^l).
\end{align*}
 In this paper,  for simplicity, we deal only with two terminals and thus two dimensional random vectors in $\lin _2$. In this case, for the two terminal source $(X,Y) \in \lin _2$ there are independent nonsingular random variables $Z^n$ and $a^n , b^n \in \mR^n$ such that $X=\sum _{i \in [n]} a_i Z_i$ and $Y=\sum _{i\in [n]} b_i Z_i$ and the joint and the conditional RID's are also well defined.

\section{Heuristic Derivation of the Multi-Terminal AMP}\label{MT_AMP_Appendix}
In this section, we try to heuristically obtain an approximation to the message passing algorithm given by equations \eqref{MT_MP_first}-\eqref{MT_MP_last}. Our derivation is similar to the heuristic derivation of the single terminal AMP in \cite{bayati}. Intuitively as the measurement matrices $A$ and $B$ are dense,  any two messages emanating from the same check node are only slightly different from each other. The same is true for the messages emanating from a variable node. 
For example, if one considers messages from check nodes to variable nodes 
\begin{align*}
r^t _{a \to i}&=u_a - \sum _{j \in [n] \backslash i} A_{a j} x^t_{j \to a}\\
&=u_a - \sum _{j \in [n]} A_{a j} x^t_{j \to a} + A_{a i} x^t_{i \to a},
\end{align*}
it is seen that for a fixed $a \in [m_x]$, $r^t_{a \to i}$ for different values of $i \in [n]$ are different because of the appearance of the last term $A_{a i} x^t_{i \to a}$ which is of  the order $O(\frac{1}{\sqrt{m_x}}) \approx O(\frac{1}{\sqrt{n}})$ as $m_x$ and $n$ are assumed to be proportional.
Similarly,  considering the messages from variable nodes to check nodes,
\begin{align*}
x^{t+1}_{i\to a}&=\eta^x_t( \sum _{b \in [m_x] \backslash a} A_{b i} r^t_{b \to i}, \sum _{c \in [m_y]} A_{c i} s^t_{c \to i})\\
&=\eta^x_t( \sum _{b \in [m_x]} A_{b i} r^t_{b \to i} - A_{a i} r^t_{a \to i}  , \sum _{c \in [m_y]} A_{c i} s^t_{c \to i}),
\end{align*}
it is observed that for a fixed $i\in [n]$, the difference of messages $x^{t+1}_{i \to a}$ for different values of $a \in [m_x]$ is again of the order $O(\frac{1}{\sqrt{n}})$.
Therefore, one gets   
\begin{align*}
r^t_{a \to i}=r^t _a + \delta r^t_{a \to i}\ ,\ s^t_{c \to k}=s^t _c + \delta s^t_{c \to k},\\
x^t_{i \to a}=x^t_i +\delta x^t_{i \to a}\ ,\  y^t_{i \to a}=y^t_i + \delta y^t_{i \to a},
\end{align*}
where the $\delta$ terms are of the order $O(\frac{1}{\sqrt{n}})$.
Replacing in Equation \eqref{MT_MP_first} and \eqref{MT_MP_second}, one obtains
\begin{align*}
r^t_a + \delta r^t_{a \to i}&=u_a -\sum_{j \in [n]} A_{a j} (x^t _j + \delta x^t_{j \to a})+ A_{a i} (x^t_i + \delta x^t _{i \to a}),\\
s^t_c + \delta s^t _{c \to k}&=v_c -\sum_{l \in [n]} B_{c l} (y^t _l + \delta y^t _{l \to c})+ B_{c k} (y^t _k + \delta^t_{k \to c}).
\end{align*}
The terms $A_{a i} \delta x^t _{i \to a}$ and $B_{c k} \delta y^t_{k \to c}$ are of the order $O(\frac{1}{n})$ and negligible asymptotically. Thus, one obtains that
\begin{align}
r^t_a&=u_a-\sum _{j \in [n]} A_{a j} (x^t_j + \delta x^t_{j \to a})\ ,\ \delta r^t_{i \to a}=A_{a i}\, x^t_i,\label{first_residual_x}\\
s^t_c&=v_c-\sum_{l \in [n]} B_{c l} (y^t_l + \delta y^t_{l \to c})\ ,\ \delta s^t_{k \to c}=B_{c k}\, y^t_k.\label{first_residual_y}
\end{align}
Replacing in Equations \eqref{MT_MP_third} and \eqref{MT_MP_last}, it results that
\begin{align*}
x^{t+1}_i&+\delta x^{t+1}_{i \to a}\\
&=\eta^x_t(\sum_{ b \in [m_x]\backslash a} A_{b i} (r^t_b + A_{b i} x^t_i) , \sum_{d \in [m_y]} B_{d i} (s^t_d + B_{d i} y^t_i))\\
&=\eta^x_t(\sum_{ b \in [m_x]} A_{b i} (r^t_b + A_{b i} x^t_i) , \sum_{d \in [m_y]} B_{d i} (s^t_d + B_{d i} y^t_i))\\
&+ \partial_1\eta^x_t(.,.) A_{a i} (r^t_a + A_{a i} x^t_i).
\end{align*}
This implies that 
\begin{align}
x^{t+1}_i&=\eta^x_t(x^t_i + \sum_{ b \in [m_x]} A_{b i} r^t_b\ ,\  y^t _i + \sum_{d \in [m_y]} B_{d i} s^t_d),\label{final_1}\\
\delta x^{t+1}_{i \to a}&=\partial_1 \eta^x_t(x^t_i + \sum_{ b \in [m_x]} A_{b i} r^t_b\ ,\  y^t _i + \sum_{d \in [m_y]} B_{d i} s^t_d) A_{a i} r^t_a,\label{second_residual_x}
\end{align}
where one uses the fact that for any $i \in [n]$, $\sum _{a \in [m_x]} A_{a i} ^2 \approx 1$, and $A_{a i} \delta r^t_{a \to i}=O(\frac{1}{n})$ thus negligible as $n$ tends to infinity. A similar argument holds for the $T_Y$giving 
\begin{align}
y^{t+1}_k&=\eta^y_t(x^t_k + \sum_{ b \in [m_x]} A_{b k} r^t_b\ ,\  y^t _k + \sum_{d \in [m_y]} B_{d k} s^t_d),\label{final_2}\\
\delta y^{t+1}_{k \to c}&=\partial_2 \eta^y_t(x^t_k + \sum_{ b \in [m_x]} A_{b k} r^t_b\ ,\  y^t _k+ \sum_{d \in [m_y]} B_{d k} s^t_d) B_{c k} r^t_c.\label{second_residual_y}
\end{align}
Replacing \eqref{second_residual_x} in \eqref{first_residual_x} and \eqref{second_residual_y} in \eqref{first_residual_y}, and using the approximation $A_{a i}^2\approx \frac{1}{m_x}$ and $B_{ c k}^2\approx \frac{1}{m_y}$, one obtains that
\begin{align}
&r^t_a=u_a - A_{a} x^t+ \frac{\sum_{j\in [n]}  \partial_1 \eta^x_t(x^{t-1} _j+ \dots, y^{t-1}_j + \dots) }{m_x} r^{t-1}_a\nonumber\\
&=u_a -  A_{a} x^t + \frac{\bracket{\partial_1 \eta^x_t(x^{t-1} + A^* r^{t-1}, y^{t-1}+B^*s^{t-1})}}{\rho_x} r^{t-1}_a, \label{final_3}
\end{align}
where $A_a$ denotes the $a$-th row of the matrix $A$. Similarly,
\begin{align}
&s^t_c=v_c - B_{c } y^t+ \frac{\sum_{l\in [n]}  \partial_2 \eta^y_t(x^{t-1} _l+ \dots, y^{t-1}_l + \dots) }{m_y} s^{t-1}_c\nonumber \\
&=v_c - B_{c} y^t + \frac{\bracket{\partial_2 \eta^y_t(x^{t-1} + A^* r^{t-1}, y^{t-1}+B^*s^{t-1})}}{\rho_y} s^{t-1}_c.\label{final_4}
\end{align}
Equations \eqref{final_1}, \eqref{final_2}, \eqref{final_3} and \eqref{final_4} give the the MAMP algorithm.

\section{Heuristic Derivation of the State Evolution}
To give an intuitive justification (as in \cite{bayati}) for the validity of SE for the two terminal AMP in Equations \eqref{se_mmse1} and \eqref{se_mmse2}, consider the following version of the AMP where at each iteration $t$, the measurement matrices $A$ and $B$ are replaced with independent copies and where we drop the Onsager term in Equations \eqref{res_x} and \eqref{res_y}. In other words, let $\vecu^t=A(t)\vecx_0 + w_x$ and $\vecv^t=B(t)\vecy_0 + w_y$ be the noisy measurements at iteration $t$, where $w_x$ and $w_y$ are additive noises consisting of i.i.d. zero mean with variance $\sigma_x^2$ and $\sigma_y^2$ respectively. The new AMP algorithm can be written as follows
\begin{align*}
&r^t=\vecu^t- A(t) x^t\ ,\  x^{t+1}=\eta^x_t(A(t)^* r^t + x^t,B(t)^* s^t + y^t),\\
&s^t=\vecv^t-B(t) y^t\ ,\  y^{t+1}=\eta^y_t(A(t)^* r^t + x^t,B(t)^* s^t + y^t).
\end{align*}
The first equation can be simplified to the following form
\begin{align*}
x^{t+1}=\eta^x_t( &\vecx_0 + A(t)^* w_x +(I-A(t)^* A(t)) (x^t-\vecx_0),\\
&\vecy_0 +B(t)^* w_y+(I-B(t)^*B(t))(y^t-\vecy_0)).
\end{align*}
Conditioned on $w_x$, $A(t)^* w_x$ is an $n$ dimensional vector with i.i.d. Gaussian components with zero mean and variance $\frac{\|w_x\|_2^2}{n} \approx \sigma_x^2$. Moreover, in the asymptotic limit as $n$ gets large, by central limit theorem, each row of $I-A(t)^*A(t)$ consists of approximately Gaussian random variables with variance $\frac{n}{m_x}=\frac{1}{\rho_x}$. Similarly, the components of $B(t)^* w_y$ are i.i.d. Gaussian with zero mean and approximate variance $\sigma_y^2$ and  each row of $I-B(t)^*B(t)$ converges to independent zero mean Gaussian  variables with variance $\frac{n}{m_y}=\frac{1}{\rho_y}$. Hence, the components of $A(t)^* w_x+ (I-A(t)^*A(t)) (x^t-\vecx_0)$ are approximately Gaussian with variance 
\begin{align}
\tau^t_x=\sigma_x^2 + \frac{1}{\rho_x} \frac{\|x^t-\vecx_0\|_2^2}{n}. \label{tau^t_x}
\end{align}
 At $t=0$, with the initialization $x^0=0$, one obtains that 
 \begin{align*}
 \tau^0_x=\sigma_x^2 + \frac{1}{\rho_x} \frac{\|\vecx_0\|_2^2}{n} \to \sigma_x^2 + \frac{1}{\rho_x} \E(X^2),
 \end{align*}
which is compatible with the SE initialization. A similar derivation gives $\tau^0_y=\sigma_y^2 + \frac{1}{\rho_y} \E(Y^2)$. Moreover, by induction on $t$, one can simply check that at iteration $t+1$, $$x^{t+1}=\eta^x_t(X+\sqrt{\tau^t_x} Z_x , Y+ \sqrt{\tau^t_y} Z_y).$$
Thus, replacing in Equation \eqref{tau^t_x} and using a similar argument, one obtains that for the iteration $t+1$, 
\begin{align*}
\frac{\|x^{t+1} - \vecx_0\|_2^2}{n} \to \E(X - \eta^x_t(X+\sqrt{\tau^t_x} Z_x , Y+ \sqrt{\tau^t_y} Z_y))^2,
\end{align*}
which implies that at iteration $t+1$:
\begin{align*}
\tau^{t+1}_x=\sigma_x^2 + \frac{1}{\rho_x} \E(X - \eta^x_t(X+\sqrt{\tau^t_x} Z_x , Y+ \sqrt{\tau^t_y} Z_y))^2.
\end{align*}
A similar argument gives the corresponding equation for $\tau_y^t$:
\begin{align*}
\tau^{t+1}_y=\sigma_y^2 + \frac{1}{\rho_y} \E(Y - \eta^y_t(X+\sqrt{\tau^t_x} Z_x , Y+ \sqrt{\tau^t_y} Z_y))^2.
\end{align*}

\section{MMSE Estimator for a Linearly Correlated Bernoulli-Gaussian Signal}\label{mmse_bg}
Suppose $Z^k$ are independent Bernoulli-Gaussian random variables with probability distribution $p_i(z)=(1-\alpha_i) \delta_0(z) + \alpha_i \normal(0, \frac{1}{\alpha_i},z)$. Let $A$ be a $t \times k$ matrix and let $S=A Z^k$ be a $t$ dimensional linearly correlated signal. Suppose $O=S+\tilde{N}$ is the observation vector, where $\tilde{N}$ is a $t\times 1$ zero mean Gaussian measurement noise with a covariance matrix $\tilde{\Sigma}$. We denote by $\eta_i(x)=\E(S_i| O=x)$ the MMSE estimator of $S_i$, the $i$-th component of the signal, given a $t\times 1$ observation vector $O=x$. We will compute $\eta_1(x)$. The other estimators can be computed similarly. 

It is easy to check that one can represent $Z_i$, $i\in[k]$ by $\Theta_i N_i$, where $\Theta^k$ are independent binary random variables with $\pp(\Theta_i=1)=\alpha_i$ and $N^k$ are independent zero mean Gaussian variables with variance $\frac{1}{\alpha_i}$. Assume that $\Sigma$ is the covariance matrix of $N^k$ with diagonal elements $\Sigma_{ii}=\frac{1}{\alpha_i}$ and zero elsewhere. Let $a_1$ denote the first row of $A$ and assume that for a given binary sequence $\theta^k$ and for an arbitrary $n \times k$ matrix $B$, $B(\theta^k)$ denotes an $n \times k$ matrix whose $i$-th column is the $i$-th column of $B$ provided $\theta_i=1$ and zero otherwise. 

Using the conditioning on $\Theta^k$, we have
\begin{align*}
\eta_1(x)&=\sum_{\theta^k \in \{0,1\}^k} \E(S_1|O=x, \Theta^k=\theta^k) \pp(\theta^k|O=x).
\end{align*}
Conditioned on $\theta^k$, $S_1=a_1(\theta^k) N$ is a zero mean Gaussian with variance $a_1(\theta^k) \Sigma\, a_1(\theta^k)^*$. The observation vector is also Gaussian with a zero mean and a covariance matrix $A(\theta^k) \Sigma\, A(\theta^k)^* + \tilde{\Sigma}$, thus the estimation of $S_1$ is reduced to a Gaussian estimation problem where the estimator is known to be a linear function of observation. 
Let $\hat{S}_1(\theta^k,x)=a(\theta^k) \Sigma A(\theta^k) ( A(\theta^k) \Sigma A(\theta^k)^* + \tilde{\Sigma})^{-1} x$. It is easy to check that $\E(S_1|O=x, \Theta^k=\theta^k)=\hat{S}_1(\theta^k,x)$. Therefore, one obtains
\begin{align*}
\eta_1(x)&=\sum_{\theta^k}  \hat{S}_1(\theta^k,x)\pp(\theta^k|O=x)\\
&=\frac{1}{p_o(x)} \sum_{\theta^k}  \hat{S}_1(\theta^k,x)\pp(\theta^k) p_o(x|\theta^k)\\
&=\frac{\sum_{\theta^k}  \hat{S}_1(\theta^k,x)\pp(\theta^k) \normal(0, A(\theta^k) \Sigma A(\theta^k)^* + \tilde{\Sigma},x)}{\sum_{\theta^k} \pp(\theta^k) \normal(0, A(\theta^k) \Sigma A(\theta^k)^* + \tilde{\Sigma},x)},
\end{align*}
where $\normal(\mu, C, x)= \frac{1}{\sqrt{(2\pi)^n \det(C)}} \exp(-\frac{1}{2} (x-\mu)^* C^{-1} (x-\mu))$ denotes the Gaussian distribution with mean $\mu$ and covariance matrix $C$.
\end{document}